\documentclass[11pt,letterpaper]{article}

\usepackage[a4paper,top=1in,bottom=1in,left=1in,right=1in,marginparwidth=1in]{geometry}
\usepackage{times}
\usepackage[symbol]{footmisc}
\usepackage{amssymb}
\usepackage{amsmath}
\usepackage{amsthm}
\usepackage{thm-restate}
\usepackage{hyperref}
\usepackage{enumerate}
\usepackage{xcolor}
\usepackage{tikz}
\usepackage[capitalize]{cleveref}
\usepackage{tcolorbox}

\def\colorschemesepia{sepia}
\def\colorschemedark{dark}
\def\colorschemelight{light}

\ifx\colorscheme\undefined
\let\colorscheme\colorschemelight
\fi

\ifx\colorscheme\colorschemelight
\colorlet{textColor}{black}
\colorlet{bgColor}{white}
\fi

\ifx\colorscheme\colorschemesepia
\definecolor{textColor}{HTML}{433423}
\definecolor{bgColor}{HTML}{fbf0da}
\fi

\ifx\colorscheme\colorschemedark
\definecolor{textColor}{HTML}{bdc1c6}
\definecolor{bgColor}{HTML}{202124}
\definecolor{textRed}{HTML}{ff968c}  
\definecolor{textGreen}{HTML}{70cc70}  
\definecolor{textBlue}{HTML}{8cbcff}  
\definecolor{textCyan}{HTML}{70cccc}  
\definecolor{textMagenta}{HTML}{d982d9}  
\definecolor{textYellow}{HTML}{bfbf69}  
\else
\colorlet{textRed}{red!50!black}
\colorlet{textGreen}{green!50!black}
\colorlet{textBlue}{blue!50!black}
\colorlet{textCyan}{cyan!80!black}
\colorlet{textMagenta}{magenta!80!black}
\colorlet{textYellow}{yellow!60!black}
\definecolor{textPurple}{HTML}{681da8}
\fi

\ifx\colorscheme\colorschemelight\else
\pagecolor{bgColor}
\color{textColor}
\fi

\hypersetup{
    hypertexnames=false,bookmarksnumbered=true,
    colorlinks,allcolors=textBlue
}

\newcommand*{\Th}{^{\textrm{th}}}
\newcommand*{\WLoG}{Without loss of generality}
\newcommand*{\wLoG}{without loss of generality}

\let\eps\varepsilon
\newcommand*{\defeq}{:=}
\newcommand*{\boolOne}{\mathbf{1}}  

\DeclareMathOperator*{\E}{\mathbb{E}}
\DeclareMathOperator*{\argmin}{argmin}

\DeclareMathOperator{\poly}{poly}
\DeclareMathOperator{\range}{range}
\DeclareMathOperator{\subsidy}{subsidy}
\DeclareMathOperator{\sortAndSplit}{\mathtt{sortAndSplit}}
\DeclareMathOperator{\condorcetSplit}{\mathtt{condorcetSplit}}

\newcommand{\bx}{\mathbf{x}}
\newcommand{\by}{\mathbf{y}}
\newcommand{\bX}{\mathbf{X}}
\newcommand{\bp}{\mathbf{p}}
\newcommand{\bs}{\mathbf{s}}

\newcommand{\N}{\mathbb{N}}

\newcommand{\cI}{\mathcal{I}}
\newcommand*{\Ical}{\mathcal{I}}
\newcommand*{\Icalhat}{\hat{\Ical}}
\newcommand*{\Ahat}{\hat{A}}
\newcommand*{\Mhat}{\hat{M}}

\newcommand*{\vhat}{\hat{v}}

\newcommand*{\betahat}{\hat{\beta}}

\newcommand{\LP}{\mathsf{LP}}
\newenvironment{proofof}[1]{{\vspace*{5pt} \noindent\bf Proof of #1:  }}{\hfill\rule{2mm}{2mm}\vspace*{5pt}}

\theoremstyle{plain}
\newtheorem{theorem}{Theorem}[section]
\newtheorem{lemma}[theorem]{Lemma}
\newtheorem{corollary}[theorem]{Corollary}

\newtheorem{observation}[theorem]{Observation}

\theoremstyle{definition}
\newtheorem{definition}[theorem]{Definition}
\newtheorem{example}[theorem]{Example}

\allowdisplaybreaks

\makeatletter
\g@addto@macro{\UrlBreaks}{%
\do\/%
\do\a\do\b\do\c\do\d\do\e\do\f\do\g\do\h\do\i\do\j\do\k\do\l\do\m%
\do\n\do\o\do\p\do\q\do\r\do\s\do\t\do\u\do\v\do\w\do\x\do\y\do\z%
\do\A\do\B\do\C\do\D\do\E\do\F\do\G\do\H\do\I\do\J\do\K\do\L\do\M%
\do\N\do\O\do\P\do\Q\do\R\do\S\do\T\do\U\do\V\do\W\do\X\do\Y\do\Z%
\do\0\do\1\do\2\do\3\do\4\do\5\do\6\do\7\do\8\do\9%
}
\makeatother

\makeatletter
\@ifpackageloaded{enumitem}{%
\newenvironment*{tightemize}{\begin{itemize}[noitemsep]}{\end{itemize}}%
\newenvironment*{tightenum}{\begin{enumerate}[noitemsep]}{\end{enumerate}}%
}{%
}
\makeatother

\let\citet\cite
\let\citep\cite

\title{Tight Subsidy Bounds for Weighted Proportional Allocation \\
of Mixed Manna}
\author{Jugal Garg\thanks{University of Illinois at Urbana-Champaign, USA. J.~Garg and E.~Sharma were supported by NSF Grant CCF-2334461.}\\\texttt{\small jugal@illinois.edu} \and Eklavya Sharma\footnotemark[1]\\\texttt{\small eklavya2@illinois.edu} \and Xiaowei Wu\thanks{University of Macau, China. Xiaowei Wu is funded by the Science and Technology Development Fund (FDCT), Macau SAR (file no. 0147/2024/RIA2, 001/2024/SKL, 0002/2025/EQP and CG2026-IOTSC), and University of Macau (file no. MYRG-GRG2025-00033-IOTSC).}\\\texttt{\small xiaoweiwu@um.edu.mo}}
\date{\empty}
\begin{document}

\maketitle

\begin{abstract}
We study the problem of fairly allocating $m$ indivisible items among $n$ agents with possibly unequal entitlements in the mixed manna setting, where each item may be perceived as a good or a chore by different agents. We focus on the fundamental fairness notion of proportionality. Since proportional allocations need not exist in this setting, we allow monetary subsidies
to restore proportionality while minimizing the total subsidy.
When each item's (dis)utility is bounded by 1, a total subsidy of at least $\tau(n) \approx n/4$ may be necessary. For goods-only or chores-only instances, the best previously known upper bound was $n/3-1/6$ due to Wu and Zhou~(2024). 
We close this gap by proving that a total subsidy of at most $\tau(n)$ always suffices, thereby establishing the tight subsidy bound. Our results hold even in the more general setting of weighted mixed manna, resolving an open question posed by~Wu et al. (2023) and Garg et al. (2026).
The allocation also satisfies weighted proportionality up to one item (WPROP1).
Our proof develops a novel application of the Knaster–Kuratowski–Mazurkiewicz (KKM) fixed-point theorem, extending the KKM framework to share-based fairness notions.
Finally, we design a polynomial-time algorithm to compute such allocations for any fixed number of agents.
\end{abstract}
\thispagestyle{empty}
\newpage
\setcounter{page}{1}
\section{Introduction}
\label{sec:intro}

Fair allocation of a set $M$ of $m$ indivisible items among $n$ agents is a fundamental problem in economics and computer science. We assume that each agent has additive preferences, where agent $i$ assigns a value $v_{i}(e)$ to item $e$ and the value of any bundle $S\subseteq M$ to agent $i$ is given by $v_i(S) = \sum_{e\in S} v_{i}(e)$.
These problems arise in various contexts, such as assigning tasks among employees or distributing environmental compliance responsibilities across companies. 
In many of these scenarios, the items to be allocated consist of both desirable goods and undesirable chores---a setting known as \emph{mixed manna}.
Furthermore, the agents might have different capacities, obligations, or asymmetric entitlements. 
For example, a hospital team with more personnel may be expected to handle a larger share of shifts, or a company with higher emissions should bear a proportionally larger portion of environmental duties. In this general framework, each agent $i$ is associated with a weight $w_i$, which represents different entitlements or responsibilities, where the total weight sums to 1.

Two central notions of fairness in the continuous setting, where items are divisible and can be fractionally allocated, are \emph{proportionality}~\cite{steihaus1948problem} and \emph{envy-freeness}~\cite{foley1967resource}, which are share-based and envy-based fairness criteria, respectively.
An allocation $\bX=(X_1, \dots, X_n)$, which is a partition of $M$ into $n$ bundles (where $X_i$ is assigned to agent $i$) is weighted proportional (WPROP) if every agent receives a bundle worth at least a fraction of their weight of the total value, i.e., $v_i(X_i) \ge w_i\cdot v_i(M)$ for all agents $i$. On the other hand, an allocation is weighted envy-free if no agent envies another agent's bundle relative to their weight, i.e., $v_i(X_i)/w_i \ge v_i(X_j)/w_j$ for all agents $i, j$.
While these fairness criteria are compelling, it is well-known that allocations satisfying them do not always exist when items are indivisible, which motivates the study of approximate fairness or the use of monetary compensation. In this paper, we focus on the latter approach for fair allocation of indivisible mixed manna, aiming to achieve weighted proportionality through monetary adjustments, referred to as \emph{subsidies}.

For the setting of symmetric agents, i.e., where agents have the same weights,
\citet{sagt/HalpernS19} initiated the study of envy-free allocation with subsidy
for indivisible goods.
An allocation with subsidy is envy-free if no agent envies the bundle plus subsidy owned by any other agent. Assuming that every good has value at most $1$ to every agent, they showed that a total subsidy of at least $n-1$ is necessary to guarantee envy-freeness. \citet{sigecom/BrustleDNSV20} established a matching upper bound by proposing an algorithm that computes envy-free allocations with subsidy at most $n-1$. An analogous result was also shown for chores by~\citet{wine/WuZZ23}, under the assumption that every chore has disutility at most $1$ to every agent.
In addition, they initiated the study of proportional allocation with subsidy, where an allocation with subsidy is proportional if the utility of every agent obtained from the assigned bundle and subsidy is at least her proportional share.
They showed that achieving proportionality requires a total subsidy of at least $\tau(n)$, where $\tau(n)=n/4$ when $n$ is even and $(n^2-1)/(4n)$ when $n$ is odd, and proposed polynomial-time algorithms that compute proportional allocations for both goods and chores using a subsidy of at most $n/4$, which is tight when $n$ is even. In a follow-up work, \citet{ijcai/WuXZ25} generalized the result to the allocation of indivisible mixed manna, and improved the upper bound to $\tau(n)$, which is tight for all $n$.

For the more general setting of asymmetric agents, the first algorithm for computing WPROP allocations with subsidy was proposed by~\citet{wine/WuZZ23}, who provided an upper bound of $\left(n-1\right)/2$ on the total subsidy.
Later, \citet{wine/WuZ24} improved this bound to $n/3-1/6$ using a \emph{tree-splitting rounding} technique. Both results apply to goods-only and chores-only instances, but do not cover general mixed manna.
Recently, \citet{garg2025wprop+PO4chores} showed that the subsidy bound of $n/3-1/6$ for chores can be achieved together with the efficiency guarantee of \emph{Pareto optimality} (PO), by rounding a competitive equilibrium. However, whether the tight guarantee of $\tau(n)$ can be obtained for weighted proportionality has remained unknown, and this question has been highlighted as an open problem in several works~\cite{wine/WuZZ23,wine/WuZ24,ijcai/WuXZ25,garg2025wprop+PO4chores}.

\subsection{Our Contributions}

In this work, we study the subsidy requirement for achieving weighted proportionality and answer the above open question in the affirmative.

\begin{tcolorbox}[colback=textColor!7!bgColor, frame empty, colupper=textColor]
{\bf Result 1.~}
For any mixed-manna instance with $n$ agents with additive valuation functions and general weights, there always exists a WPROP1 allocation that can be made WPROP with subsidy at most $\tau(n)$.
\end{tcolorbox}

Our work not only improves the subsidy bounds established in previous works  
for goods or chores, but also shows a tight guarantee for the more general setting of mixed manna.
Prior to our work, the same tight bound was known only for the unweighted case~\cite{ijcai/WuXZ25}.
Our bound depends only on the number of agents and is independent of their entitlements. This differs from weighted envy-freeness with subsidy, for which the bounds for general additive goods in~\cite{corr/abs-2502-09006} depend on the entitlements.
The allocation simultaneously satisfies WPROP1; for canonical instances, it is also fPO. The reduction to general instances preserves WPROP1 and the subsidy guarantee, but does not establish PO or fPO.

\smallskip

Our analysis is based on a novel application of the Knaster–Kuratowski–Mazurkiewicz (KKM) fixed point theorem to the allocation of indivisible items. This framework has recently demonstrated strong potential in establishing the existence of near envy-free and PO allocations~\cite{IgarshiM26,corr/abs-2507-09544,corr/abs-2509-18673,corr/abs-2507-03946}. In this work, we extend the framework for the first time to handle share-based fairness notions.
More broadly, these recent developments establish the KKM theorem as a powerful new tool for discrete fair division. At the same time, applying the KKM theorem presents its own technical challenges, which require problem-specific construction and arguments. We hope that our work will inspire further applications of this framework to other challenging open problems.

Below, we provide an overview of our analysis framework and highlight some key challenges.
Our analysis consists of three main components:
\begin{itemize}
    \setlength\itemsep{0.5em}
    \item[(1)] We first perturb the input instance to obtain some desired properties including non-degeneracy.
    Then we reduce the problem to computing certain optimal integral solutions of a linear program $\LP(\beta)$ defined on the perturbed instance, where the LP
    aims at maximizing the $\beta$-weighted social welfare, parameterized by a coefficient vector $\beta \in \Delta_{n-1}$ in the $(n-1)$-dimensional simplex.
    \item[(2)] We introduce a key concept called ``$i$-biased'' (analogous to \emph{price envy-freeness} in previous works) for vectors in $\Delta_{n-1}$, and define the closed subsets $\{C_1,C_2,\ldots,C_n\}$ where $C_i$ contains all $i$-biased vectors.
    By showing the covering condition of the subsets, we use the KKM theorem to show the existence of a vector $\beta\in\Delta_{n-1}$ that is $i$-biased for all agents $i\in N$.
    \item[(3)] For this particular vector $\beta$, we show that the linear program $\LP(\beta)$ admits an optimal integral solution requiring total subsidy at most $\tau(n)$ to achieve weighted proportionality.
\end{itemize}

The first step (reduction to non-degenerate instances) is carried out through a standard random perturbation procedure on the valuation functions, as in previous KKM-based works. However, several subtle yet crucial details arise in designing the perturbation. As noted by~\citet{corr/abs-2509-18673}, a major difficulty in applying the KKM theorem to mixed manna is handling items with zero values. 
In their algorithm, items with zero values are not perturbed.
To satisfy the covering condition of the KKM theorem, one needs to ensure that for every $\beta\in\Delta_{n-1}$, there exists an agent $i$ with $\beta_i>0$ who is price envy-free under some optimal solution of $\LP(\beta)$.
However, there are simple instances for which some $\beta$ is not covered. To circumvent this issue, \citet{corr/abs-2509-18673} introduce an auxiliary item that has an identical small positive value for all agents.
However, this approach does not work in our setting: for weighted proportionality, we need to measure the additive distances between an agent’s price and her proportional price, and this measure is insensitive to perturbation by a single item.

Specifically, we define the notion of an ``$i$-biased'' vector $\beta$ based on
\begin{equation*}
    d_i = w_i \cdot v_i(M) - v_i(X_i),
\end{equation*}
which is the \emph{distance} of agent $i$ towards achieving her weighted proportional share under allocation $\bX$.
If $d_i\le 0$, the allocation is WPROP to agent $i$; otherwise, agent $i$ must be subsidized by $d_i$ to be WPROP.
Therefore, the overall goal is to minimize the sum of $\max\{d_i,0\}$.
For every agent $i$, the smaller $d_i$ is, the higher utility agent~$i$ receives.
Analogous to the price envy-freeness used in previous works~\cite{corr/abs-2507-03946,corr/abs-2507-09544,corr/abs-2509-18673}, we (inaccurately speaking) call a vector $\beta$ ``$i$-biased'' if there exists an optimal allocation to $\LP(\beta)$ such that the \emph{price distance} $\beta_i\cdot d_i$ is the smallest among all agents.
To apply the KKM theorem, we must show that for every $\beta\in\Delta_{n-1}$, there exists an agent $i$ with $\beta_i>0$ such that $\beta$ is $i$-biased. Unfortunately, zero-valued items break this condition, and the auxiliary-item trick of~\citet{corr/abs-2509-18673} does not help as its perturbation to $d_i$ is negligible.
Therefore, rather than perturbing only non-zero values (as in~\citet{corr/abs-2509-18673}), we perturb \emph{all} values, including zeros.\footnote{Note that \citet{corr/abs-2507-03946} also perturbed zeros, but they set the perturbation sufficiently small so that if an agent receives an item, the original value the agent has on the item is non-zero.
Our reduction permits such allocations and preserves the subsidy and WPROP1 guarantees; it does not establish PO for the original instance.}
Furthermore, we modify the perturbation process to ensure an additional \emph{consistency} property (see Section~\ref{sec:redn-canon}), which is essential for establishing the KKM covering condition in our second step.

Using the KKM theorem, we show the existence of a vector $\beta$ that is $i$-biased for all agents, which defines the collection of feasible allocations. Utilizing the fact that $\beta$ is $i$-biased for all agents, we apply the augmenting tree technique~\cite{corr/abs-2507-09544,corr/abs-2509-18673} to show that there exists an allocation under which the differences in price distances ($\max_{i,j}|\beta_i\cdot d_i - \beta_j\cdot d_j|$) is upper bounded by the absolute price of some item, and the total price distances is non-positive.

A further difficulty arises when translating bounds on price differences to bounds on subsidy. To establish a subsidy bound, we must upper bound $\sum_{i\in N} \max\{d_i,0\}$, which is the sum of positive distances.
Existing works~\cite{wine/WuZZ23,ijcai/WuXZ25} obtain the optimal bound $\tau(n)$ in the unweighted case by showing that (1) the sum of distances $\sum_{i\in N} d_i$ is non-positive, and (2) the pairwise difference satisfies $|d_i-d_j|\le 1$ for all $i,j\in N$.
However, we can establish similar properties only for price distances $\beta_i\cdot d_i$, which do not immediately imply properties for the value distances $d_i$ (see Section~\ref{ssec:total_subsidy} for a detailed discussion). Instead, we express the subsidy bound as a concave function of the distances and apply Jensen’s inequality to obtain the required upper bound.

\smallskip

In addition to the existential result, we also establish some computational results.

\begin{tcolorbox}[colback=textColor!7!bgColor, frame empty, colupper=textColor]
{\bf Result 2.~}
For mixed manna with $n$ agents and additive valuations, when $n$ is a fixed constant, we can compute in polynomial time a WPROP allocation with total subsidy at most $\tau(n)$.
\end{tcolorbox}

A key property enabling computation is that the allocation identified in Step~(3) is fractionally Pareto optimal (fPO), which is implied by its optimality for $\LP(\beta)$. Thus, for perturbed instances, there always exists a WPROP and fPO allocation with total subsidy at most $\tau(n)$.
As in previous works, we show that the number of fPO allocations is at most $(m+1)^{\binom{n}{2}}$, and so, when the number of agents is a constant, we can search through all fPO allocations in polynomial time to find an allocation with subsidy at most $\tau(n)$.
Note that although compatibility between WPROP-with-low-subsidy and PO is trivial
(since any Pareto-improvement preserves or decreases the subsidy required),
demonstrating compatibility between WPROP-with-low-subsidy and fPO is not easy,
and this compatibility is crucial for our algorithm.

We can extend this idea to the setting of constantly many \emph{types} of agents. Specifically, agents belong to $k$ groups, where $k$ is a constant, and agents in the same group have the same valuation function. For every $\varepsilon>0$, we obtain subsidy at most $\tau(n) + \tau(k)+\varepsilon$ using a two-stage allocation algorithm, with running time polynomial in the input length and $\log\left(1+1/\varepsilon\right)$ for fixed $k$. In the first stage, we allocate the items among the $k$ groups using the fPO-enumeration-based algorithm. Then we further divide the items among agents within groups using a greedy algorithm. When $n$ is large, we have $\tau(n) \gg \tau(k)$, and so, we get a near-optimal bound on subsidy in polynomial time.

\subsection{Additional Related Work}
Due to the vast literature on the fair allocation of indivisible items, in the following, we only review the most relevant works.
For a more comprehensive review, please refer to the recent surveys~\cite{ai/AmanatidisABFLMVW23,sigecom/AzizLMW22,jair/LiuLSW24,ipl/Suksompong25}.
\medskip

\noindent\textbf{EF1 and PO Allocation.}
The envy-freeness up to one item (EF1) is a popular relaxation of envy-freeness.
EF1 and PO allocations for goods are known to exist~\cite{teco/CaragiannisKMPS19} and can be computed in pseudo-polynomial time~\cite{Barman18FFEA}.
For the allocation of chores, EF1 and PO allocations are shown to exist in a recent breakthrough result by~\citet{corr/abs-2507-09544} using the KKM fixed point theorem, who also presented polynomial-time algorithms for computing such allocations for a fixed number of agents, using enumeration-based techniques.
The existence of near envy-free and PO allocations has also been studied for the mixed manna~\cite{corr/abs-2509-18673,corr/abs-2507-03946} recently.
However, whether EF1 and PO allocations always exist for mixed manna remains open.
\medskip

\noindent\textbf{Weighted EF1 Allocations.}
Allocations that are WEF1 and PO are known to exist when all items are either goods~\cite{Barman18FFEA,teco/ChakrabortyISZ21} or chores~\cite{corr/abs-2507-09544}.
Without the PO requirement, WEF1 allocations can be computed efficiently for goods~\cite{teco/ChakrabortyISZ21} and chores~\cite{ai/WuZZ25,aaai/SpringerHY24}.
While pseudopolynomial-time algorithms exist for computing WEF1 and PO allocations for goods, this computation problem remains open for chores, except for a fixed number of agents, using enumeration-based techniques.
For mixed manna, the existence of unweighted EF1 allocations is known~\cite{aziz2021fair,bhaskar2021approximate}. 
Very recently, two independent works~\cite{teh2026weightedfairdivisionindivisible,lin2026envyfreenessadditivemixedmanna} establish polynomial-time computation of WEF1 allocations for arbitrary positive entitlements.
\medskip

\noindent\textbf{Proportionality and Its Relaxations.}
Weighted proportionality up to one item (WPROP1) is a well-studied relaxation of proportionality.
\citet{orl/AzizMS20} showed that WPROP1 and PO allocations always exist and can be computed efficiently for mixed manna. For the stronger notion of weighted proportionality up to any item (WPROPX), \citet{geb/AzizB22} showed that WPROPX allocations are not guaranteed to exist for goods, even when all agents have equal weights. In contrast, WPROPX allocations for chores can be computed efficiently~\cite{ai/AzizLMWZ24}.
\medskip

\noindent\textbf{Allocations with Subsidy.}
Besides proportionality, allocations with subsidy have also been studied under other fairness notions, including weighted envy-freeness~\cite{corr/DaiCWXZ24,ifaamas/0001HKS0SY25,ifaamas/ElmalemGS25,corr/abs-2502-09006}, maximin share fairness~\cite{ijcai/WuX025}, and equitability~\cite{aaai/000121,corr/abs-2505-23251}.
Envy-free allocations with subsidy have also been studied for graph allocations~\cite{corr/LiSSX25}, online settings~\cite{corr/abs-2510-13633} and beyond additive functions~\cite{ijcai/BarmanKNS22,ai/KawaseMSTY25}.
There are also works that focus on the computational complexity for allocations with minimum subsidy, e.g., see~\cite{wine/CaragiannisI21,orl/ChooLSTZ24}.


\section{Preliminaries}
\label{sec:prelims}

For any natural number $t \in \N$, define $[t] \defeq \{1, 2, \ldots, t\}$
and $\Delta_{t-1} \defeq \{\bx \in \mathbb{R}_{\ge 0}^t: \sum_{j=1}^t x_j = 1\}$.

A fair division instance $\Ical$ is given by a tuple $(N, M, v, w)$.
Here $M = \{e_1,\ldots,e_m\}$ denotes the set of items and $N = \{ 1,2,\ldots,n \}$ denotes the set of agents.
Each item $e$ has value $v_i(e)$ to agent $i$, which may be positive, negative, or zero.
Agents' valuation functions are additive, i.e., for any set of items $S \subseteq M$, we have
$v_i(S) \defeq \sum_{e \in S} v_i(e)$.
We say that an item $e \in M$ is a \emph{good} for agent $i$ if $v_i(e) \geq 0$, and a \emph{chore} otherwise.
We consider the weighted setting, where each agent $i$ has a weight $w_i > 0$ and $\sum_{i\in N} w_i = 1$.
Our goal is to find an allocation $\bX = (X_1,\ldots,X_n)$ that achieves weighted proportionality with a small total subsidy. For canonical instances, we additionally guarantee fractional Pareto optimality. For convenience, given a set of items $S\subseteq M$ and an item $e\in M$, we write $S+e$ and $S-e$ to denote $S\cup \{e\}$ and $S\setminus \{e\}$, respectively.
For each $i\in N$, we call $w_i\cdot v_i(M)$ the \emph{weighted proportional share} of agent $i$.

\begin{definition}[WPROP1]
    An allocation $\bX=(X_1,X_2,\ldots,X_n)$ is weighted proportional up to one item (WPROP1), if for every agent $i\in N$,
    either $v_i(X_i)\geq w_i\cdot v_i(M)$, or
    there exists an item $e\notin X_i$ such that $v_i(X_i + e) \geq w_i\cdot v_i(M)$,
    or there exists an item $e\in X_i$ such that $v_i(X_i - e) \geq w_i\cdot v_i(M)$.
\end{definition}

\begin{definition}[WPROP with Subsidy]
    An allocation $\bX=(X_1,X_2,\ldots,X_n)$ with non-negative subsidies $\bs=(s_1,s_2,\ldots,s_n)$ is weighted proportional (WPROP), if for every agent $i\in N$, we have
    \footnote{Another natural definition includes the total subsidy in the proportional benchmark, requiring $v_i\left(X_i\right)+s_i\geq w_i\cdot\left(v_i\left(M\right)+\sum_j s_j\right)$. Weighted envy-freeness with subsidy implies these inequalities: multiply each normalized envy inequality by $w_j$ and sum over $j$. With identical valuations and equal weights, the two notions coincide. In particular, for the unweighted goods-only and chores-only settings, the known envy-free subsidy upper bounds and the lower-bound examples in~\cite{wine/WuZZ23} give the tight bound $n-1$ under this alternative definition, rather than $\tau\left(n\right)$.}
    \begin{equation*}
        v_i(X_i) + s_i \geq w_i\cdot v_i(M).
    \end{equation*}
    We call $(\bX,\bs)$ a WPROP outcome, where $\sum_{i\in N} s_i$ is the total subsidy.
\end{definition}

Note that every allocation $\bX$ can be turned into a WPROP outcome by giving to each agent $i$ a subsidy of
$s_i = (w_i\cdot v_i(M) - v_i(X_i))^+$, where $(x)^+ = \max\{x,0\}$.
Therefore, our goal is to compute an allocation $\bX$ for which the total subsidy
\[ \subsidy(\bX,\cI) \defeq \sum_{i\in N} (w_i\cdot v_i(M) - v_i(X_i))^+ \]
is small.
For ease of notation, we introduce a function $\tau(n)$ defined as:
\[
\tau(n) =
\begin{cases}
\frac{n}{4} & \text{if } n \text{ is even}, \\
\frac{n^2 - 1}{4n} & \text{if } n \text{ is odd}.
\end{cases}
\]

We define the \emph{range} of a fair division instance to be the maximum absolute value any agent can have for any item.
Formally, for nonempty $M$, we have $\range(\Ical) \defeq \max_{i \in N, e \in M} \{|v_i(e)|\}$; for $M=\varnothing$, set $\range(\Ical)=0$.
If $n=1$, assigning all items to the single agent requires no subsidy. If $\range(\Ical)=0$, any complete allocation is WPROP and WPROP1 without subsidy. We handle these cases directly before normalizing.
We often assume \wLoG{} that $\range(\Ical) = 1$, since we can scale all valuations by the same factor.
It was shown by~\citet{wine/WuZZ23} that it requires a total subsidy of at least $\tau(n)$ to guarantee WPROP outcomes, even when all $n$ agents have identical valuation functions and equal weights.

Next, we introduce the efficiency measurements.

\begin{definition}[PO]
    An allocation $\bX=(X_1,X_2,\ldots,X_n)$ is Pareto dominated by another allocation $\mathbf{Y}=(Y_1,Y_2,\ldots,Y_n)$ if $v_i(Y_i) \geq v_i(X_i)$ for all $i\in N$ and the inequality is strict for some $i$.
    We call an allocation Pareto optimal (PO) if it is not Pareto dominated by any allocation.
\end{definition}

We also define the stronger efficiency notion of fPO.
We use $\mathbf{x} = (x_1, x_2,\ldots, x_n)$ to denote a fractional allocation, where $x_i(e)$ is the fraction of item $e$ allocated to agent $i$. The allocation satisfies $x_i(e) \geq 0$ for all $i\in N$ and $e\in M$, and $\sum_{i\in N} x_i(e) = 1$ for all $e\in M$.

\begin{definition}[fPO]
    We call an allocation $\bX=(X_1,X_2,\ldots,X_n)$ fractionally Pareto optimal (fPO) if it is not Pareto dominated by any fractional allocation.
\end{definition}

Clearly, every fPO allocation is PO, but the converse is not true.


\section{Reduction to Canonical Instances}
\label{sec:redn-canon}

In this section, we introduce some properties of the instances that are desirable for the analysis later. We show that, for the purpose of computing allocations with subsidy, it is without loss of generality (w.l.o.g.) to assume that all these properties hold.

\begin{definition}[Non-degeneracy]
\label{defn:degen}
    We call an instance $(N,M,v,w)$ \emph{non-degenerate} if for every simple alternating cycle $\left(i_0,e_1,i_1,e_2,\ldots,e_k,i_k\right)$, where $k\geq2$, $i_k=i_0$, the agents $i_0,\ldots,i_{k-1}$ are distinct, and the items $e_1,\ldots,e_k$ are distinct, we have
    \begin{equation*}
        \prod_{t=1}^k v_{i_{t-1}}(e_t) \neq \prod_{t=1}^k v_{i_t}(e_t).
    \end{equation*}
\end{definition}

\begin{definition}[Objectiveness and Strict Objectiveness]
An instance of mixed manna is \emph{objective} if every item is either a good for all agents or a chore for all agents, i.e.,
\[
M=M^+\cup M^-, \text{ where  }
M^+=\{e:\forall i,\; v_i(e)\ge0\}\  \text{ and } \ 
M^-=\{e:\forall i,\; v_i(e)<0\}.
\]
It is \emph{strictly objective} if 
$v_i(e)>0$ for every agent $i$ and every good $e\in M^+$.
\end{definition}

\begin{definition}[Consistency]
    We call an instance of mixed manna \emph{consistent}, if we have either $v_i(M) > 0$ for all $i\in N$, or  $v_i(M) < 0$ for all $i\in N$.
\end{definition}

We call an instance \emph{canonical} if it is non-degenerate, strictly objective, and consistent.
We show that for the purpose of computing a WPROP outcome with a small total subsidy, it is w.l.o.g. to assume that the instance is canonical.
Intuitively, if an item $e$ satisfies $v_i(e) > 0$ but $v_j(e) < 0$ for two agents $i,j\in N$, then we set $v_j(e) = 0$ in the modified instance, since an allocation assigning $e$ to $j$ can be improved by transferring it to $i$.
After removing redundant items and obtaining objectiveness, we can remove agents whose total value is negative whenever some agent has a nonnegative total value. The removed agents are proportional with empty bundles, and rescaling the remaining weights does not decrease their nonnegative proportional benchmarks.
We can also assume non-degeneracy and strict objectiveness by adding infinitesimal positive random noises to the value of every agent on every item.
Similar observations and reduction techniques have also been used
in existing works, e.g., see
\citet[Section 2]{corr/abs-2507-09544},
\citet[Lemma 3.4]{corr/abs-2507-03946}, \citet[Section 4]{corr/abs-2509-18673}, and \citet[Section 4.2]{ijcai/WuXZ25}.
In the following, we show the details of the reduction for completeness. 


Let $\cI = (N,M,v,w)$ be an arbitrary instance of mixed manna.
After handling the one-agent and zero-range cases, normalize $\range\left(\cI\right)=1$.
We construct a new instance $\cI' = (N', M', v', w')$, and show that it is canonical.
Our construction has four steps.
\medskip

\noindent\textbf{Step 1 (Remove Redundant Items).}
First, if there exists an item $e$ such that $v_i(e) = 0$ for some agent $i\in N$ and $v_j(e)\leq 0$ for all other agents $j\in N$, then we set this item aside.
Intuitively speaking, such items should only be allocated to agent $i$ if one tries to compute PO allocations, or minimize subsidy.
Therefore, we can focus on the allocation of other items, and then add these items back (e.g., by allocating the item $e$ to agent $i$) without breaking the Pareto optimality or subsidy bound.
Let $M'\subseteq M$ be the remaining items after this step.
If $M'=\varnothing$, allocate each removed item to an agent valuing it at zero and stop; the resulting allocation is WPROP without subsidy.
\medskip

\noindent\textbf{Step 2 (Obtain Objectiveness).}
Next, for every item $e\in M'$, if there exists an agent $i$ with $v_i(e) > 0$ (i.e., agent $i$ considers item $e$ as a good), then for all $j\in N$ such that $v_j(e) < 0$, we artificially set $v_j(e) = 0$.
Let $\hat{v}$ be the resulting valuation functions.
After this modification, for all $e\in M'$, we have either $\hat{v}_i(e)\geq 0$ for all $i\in N$ (i.e., item $e$ is a good to all agents), or $\hat{v}_i(e) < 0$ for all $i\in N$ (i.e., item $e$ is a chore to all agents).
Note that the instance is objective but not strictly objective, as some items could have a value of $0$ to some agents.
\medskip

\noindent\textbf{Step 3 (Remove Inconsistent Agents).}
If there exists an agent $i\in N$ with $\hat{v}_i(M') \geq 0$, then we remove all agents $j$ with $\hat{v}_j(M') < 0$, and rescale the weights of the remaining agents by a common factor so that they sum to $1$.
Let $N'\subseteq N$ be the set of remaining agents.
After this modification, we have either $\hat{v}_i(M') \geq 0$ for all $i\in N'$, or $\hat{v}_i(M') < 0$ for all $i\in N'$.
Note that the instance is not consistent yet, as it is possible that $\hat{v}_i(M') = 0$ holds for some agent $i$.
\medskip

\noindent\textbf{Step 4 (Random Perturbation).}
Finally, we add infinitesimal positive random noises to the valuation functions to achieve non-degeneracy and eliminate zero values.
Let
\begin{equation*}
    \mathcal{P} = \left\{ (P_i)_{i \in N'} : \forall i\neq j,
    P_i\cap P_j=\varnothing, \bigcup_{i\in N'} P_i \subseteq M' \right\}
\end{equation*}
be the set of all $|N'|$-tuples of
disjoint bundles (which are not necessarily a complete partition of $M'$).
We define two (positive and small) parameters $\lambda$ and $\eta$ satisfying the following:
\begin{align*}
    \lambda & < \min\left\{ \left|v_i(S) - w_i\cdot v_i(M)\right|:
        i\in N, \; S\subseteq M,\; v_i(S) \neq w_i\cdot v_i(M) \right\},
    \\ \lambda & < \min\left\{ \sum_{i\in N'} \hat{v}_i(A_i) - \sum_{i\in N'} \hat{v}_i(B_i):
        A, B \in \mathcal{P}, \sum_{i\in N'} \hat{v}_i(A_i) > \sum_{i\in N'} \hat{v}_i(B_i)  \right\},
    \\ \eta & = \frac{\lambda}{m}.
\end{align*}
Here a minimum over an empty set is interpreted as $+\infty$. The first bound preserves both positive and negative proportionality gaps; the second preserves nonzero item values and total values.
Then for all $i\in N'$ and $e\in M'$, we let $v'_i(e) = \hat{v}_i(e) + \eta_{ie}$, where $\eta_{ie}\in (0,\eta)$ is an independent random variable drawn uniformly from $(0, \eta)$.
Note that after the perturbations, we might have $v'_i(e) > 1$.
However, we still have $v'_i(e) \in [-1, 1+\eta)$, and this difference in value domain is negligible (in terms of subsidy bound) when $\eta$ is sufficiently small.

\medskip

Next, we show that the constructed instance is canonical, and fair allocations for the canonical instance can be transformed into fair allocations for the original instances.
The proofs of Lemmas~\ref{lemma:constructed_instance_is_canonical} and~\ref{lemma:reduction_to_canonical_instances} are deferred to Appendix~\ref{sec:missing_proofs}.

\begin{lemma} \label{lemma:constructed_instance_is_canonical}
    Given any instance $\cI$ of mixed manna, the instance $\cI'$ constructed above is canonical with probability one, provided that $M'\neq\varnothing$.
\end{lemma}

More importantly, we show that the construction is subsidy-preserving in the sense that any allocation for the perturbed instance $\cI'$ can be transformed into an allocation for the original instance $\cI$ with (almost) the same subsidy requirement.

\begin{lemma} \label{lemma:reduction_to_canonical_instances}
    Given any allocation $\bX'$ for instance $\cI'$, we can construct in polynomial time an allocation $\bX$ for instance $\cI$, such that
    \begin{equation*}
        \subsidy(\bX,\cI) \leq \subsidy(\bX',\cI') + \lambda.
    \end{equation*}
    Moreover, if \,$\bX'$ is WPROP1 for instance $\cI'$, then $\bX$ is also WPROP1 for instance $\cI$.
\end{lemma}

Note that $\lambda$ can be chosen to be arbitrarily small,
and so, the construction can increase the subsidy by only a negligible amount (see Section~\ref{ssec:total_subsidy} for the details).

It is easy to see that the first three steps of constructing $\Ical'$ can be performed in polynomial time.
The fourth step, where we randomly perturb values, can also be performed in polynomial time when
the number of agents is a constant, using a trick from Section 4 of \citet{corr/abs-2507-03946}.
The key idea is that instead of picking $\eta_{ie}$ independently and uniformly randomly from $(0, \eta)$,
we choose each $\eta_{ie}$ from the grid
\begin{equation*}
    \left\{\frac{\eta\cdot t}{2\cdot\left(m\cdot n\right)^n+2}:t\in\left\{1,\ldots,2\cdot\left(m\cdot n\right)^n+1\right\}\right\}\subset\left(0,\eta\right).
\end{equation*}
There are at most $2\cdot\left(m\cdot n\right)^n$ simple-cycle equations in \cref{defn:degen}. Assign the perturbations one at a time. When assigning the last variable of a cycle, its equation is linear in that variable with a nonzero coefficient, since all other perturbed values are nonzero. Hence each completed cycle forbids at most one grid value, and a permitted value always remains. This yields a deterministic construction in polynomial time for fixed $n$, provided that $\lambda$ has polynomial encoding length; suitable rational choices are given in \cref{sec:fpo-enum}.
Given the above reduction, it suffices to show that for every canonical instance, there exists an allocation that requires a small amount of subsidy.
In fact, we show (in \cref{thm:canonical-existence}) the existence of such allocations that are also fPO%
\footnote{This is important as we can enumerate all fPO allocations for a constant number of agents (see Section~\ref{sec:few-types}).}%
, via a parameterized optimization problem. This fPO guarantee is for the canonical instance; the reduction does not establish PO or fPO for the original instance.
In the following, unless otherwise specified, we assume that the instance is canonical.

\section{Allocations with Tight Subsidy Guarantee}
\label{sec:tight-subsidy}

In this section, we show the existence of WPROP and fPO allocations with a small total subsidy for the canonical instances via the KKM fixed point theorem.

\subsection{The Parameterized Allocation Program}

Consider the following LP for a weighted allocation problem (parameterized by coefficient vector\footnote{In previous works, the coefficient vector was referred to as the weight vector and denoted by $w$. To distinguish this vector from the weights of agents we study in this paper, we use $\beta$ to denote this vector, and call it a coefficient vector.} $\beta = (\beta_1,\beta_2,\ldots,\beta_n)$, on instance $\cI$) and its dual, where $\mathbf{x}$ is a fractional allocation where $x_{i}(e)$ denotes the fraction of item $e$ allocated to agent $i$.
Note that the LPs are independent of the weights of agents.
\begin{align*}
    \LP(\beta) &  && \qquad\qquad & \mathsf{Dual(\beta)} &  \\
    \max. \quad &  \sum_{i\in N} \sum_{e\in M} ~ \beta_i\cdot v_{i}(e)\cdot x_{i}(e) && \qquad\qquad & \min. \quad & \sum_{e \in M} p(e) \\
    \text{s.t.} \quad &  \sum_{i\in N} x_{i}(e) = 1 && \forall e\in M & \text{s.t.} \quad & p(e) \geq \beta_i\cdot v_{i}(e) && \forall i\in N,e\in M \\
    &  x_{i}(e) \geq 0 && \forall i\in N,e\in M & &  &&
\end{align*}

We say that vector $\beta > 0$ if $\beta_i > 0$ holds for all $i\in N$.

\begin{observation}\label{observation:unique_dual}
    The optimal dual solution is unique: for all $e\in M$, we have
    \begin{equation*}
        p(e) = \max_{i\in N} \left\{ \beta_i\cdot v_i(e) \right\},
    \end{equation*}
    which we call the \emph{price} of item $e$.
    Moreover, when $\beta > 0$, the prices of all goods are positive; the prices of all chores are negative.
\end{observation}
\begin{proof}
    Clearly, setting $p(e) = \max_{i\in N}\{ \beta_i\cdot v_i(e) \}$ for every item $e\in M$ gives the unique optimal solution to the dual LP.
    Suppose $\beta > 0$.
    If item $e$ is a good, then we have $p(e) = \max_{i\in N}\{ \beta_i\cdot v_i(e) \} > 0$ because $\beta_i\cdot v_i(e)$ is positive for all $i\in N$.
    Similarly, if $e$ is a chore, then we have $p(e) = \max_{i\in N}\{ \beta_i\cdot v_i(e) \} < 0$ because all $\beta_i\cdot v_i(e)$ are negative.
\end{proof}

For all $S\subseteq M$, we use $p(S) = \sum_{e\in S} p(e)$ to denote the total price of items in $S$.
By dual feasibility, for every item $e\in M$ and agent $i\in N$, we have $p(e) \geq \beta_i\cdot v_i(e)$.
We define the items for which this inequality is tight as the MBB items of agent $i$.

\begin{definition}[MBB Items]
    We define the maximum-bang-per-buck (MBB) items of agent $i$ as
    \begin{equation*}
        M_i = \left\{ e\in M: p(e) = \beta_i\cdot v_i(e) \right\}.
    \end{equation*}
\end{definition}

Note that once we fix $\beta$, the prices of items are fixed, so are the MBB sets of agents.
We call a feasible solution $\bx$ to the primal LP \emph{MBB-feasible} if $x_{i}(e)>0$ holds only when $e\in M_i$.
That is, every agent only receives (a subset of) their MBB items.
Note that every feasible solution corresponds to a fractional allocation.
When the solution is integral, i.e., $x_{i}(e)\in \{0,1\}$ for all $i\in N$ and $e\in M$, we call it an allocation, and use $\bX = (X_1,\ldots,X_n)$ to denote it, where $X_i = \{ e\in M : x_{i}(e)=1 \}$.

\begin{observation} \label{observation:complementary_slackness}
    By complementary slackness, any feasible solution $\bx$ to $\LP(\beta)$ is optimal if and only if it is MBB-feasible.
\end{observation}

The above observation implies that the primal LP has at least one integral optimal solution: given an arbitrary feasible solution $\bx$ that is MBB-feasible, an arbitrary rounding that preserves LP-feasibility and MBB-feasibility gives an integral optimal solution (e.g., by allocating each item $e$ to the agent $i$ with maximum $x_i(e)$).
Furthermore, we show that when $\beta > 0$, every MBB-feasible allocation is fPO.

\begin{lemma} \label{lemma:MBB-feasible_implies_PO_when_beta>0}
    If the coefficient vector $\beta > 0$, then every MBB-feasible solution $\bx$ is fPO.
\end{lemma}
\begin{proof}
    By Observation~\ref{observation:complementary_slackness}, $\bx$ is optimal for $\LP(\beta)$.
    If $\bx$ is dominated by another fractional allocation $\mathbf{y}$ (which defines a feasible solution to $\LP(\beta)$), then since $\beta_i > 0$ for all $i\in N$, the objective of $\mathbf{y}$ is strictly larger than that of $\bx$, which contradicts $\bx$ being optimal.
\end{proof}

In summary, if we have $\beta > 0$ (which defines the non-zero prices of items), then we only need to focus on computing an integral allocation $\bX$ that is MBB-feasible (which implies fPO), and satisfies certain fairness guarantees.
Let $\LP^*(\beta)$ be the collection of optimal integral solutions to $\LP(\beta)$, which is also the collection of all MBB-feasible allocations.

Assuming that $\sum_{i\in N}\beta_i=1$ and $\beta\geq 0$, all possible coefficient vectors form the simplex $\Delta_{n-1}$.
However, as discussed above, to ensure fPO, we are only interested in $\beta > 0$.
Therefore, as in existing works~\cite{corr/abs-2509-18673,corr/abs-2507-09544,corr/abs-2507-03946}, we focus on a subset of $\Delta_{n-1}$ in which all vectors are positive.
We introduce another (positive and small) parameter
\begin{equation*}
    \epsilon < \frac{\min\left\{ \min_{i\in N, e\in M} \{ w_i\cdot |v_i(e)| \}, \min_{i\in N}\{ w_i\cdot |v_i(M)| \} \right\} }{nm\cdot \max_{i\in N, e\in M} \{ |v_i(e)| \}}.
\end{equation*}

Intuitively speaking, we need $\epsilon$ to be positive but also sufficiently small so that $\hat{\beta}$ (to be defined next) and $\beta$ are close to each other.
The reason for setting this (complex) form for $\epsilon$ will be clear in the proof of Lemma~\ref{lemma:beta_i>0_and_beta_in_C_i} (showing the covering condition for the KKM theorem).

Then we relate each vector $\beta\in \Delta_{n-1}$ with
\begin{equation*}
    \hat{\beta} = \left( (1-n\epsilon)\beta_1+\epsilon,\; (1-n\epsilon)\beta_2+\epsilon,\; \ldots,\; (1-n\epsilon)\beta_n+\epsilon \right).
\end{equation*}
Note that $\hat{\beta}\in \Delta_{n-1}$ and $\hat{\beta}>0$.
Intuitively speaking, $\hat{\beta}$ is obtained by shrinking $\beta$ toward the center of the simplex by a factor of $(1-n\epsilon)$, and these vectors form a slightly shrunken version of $\Delta_{n-1}$.
The vector $\hat{\beta}$ is only a technical device for the KKM argument in Section~\ref{ssec:fix_beta}: it lets us use the closed simplex while evaluating strictly positive coefficient vectors. After Lemma~\ref{lemma:beta_in_intersection_of_C}, we work directly with the resulting positive vector $\beta$, and $\hat{\beta}$ plays no further role.
\medskip

\noindent\textbf{Roadmap.}
Note that the LPs depend only on the coefficient vector $\beta$, and not on the weights of agents.
The coefficient vector $\beta$ uniquely determines the price of items, and also the MBB item sets and the MBB-feasibility graph (to be defined later).
By carefully choosing $\beta$ (taking the weights of agents into consideration), we have some desired property regarding the MBB feasibility.
We also take the weights of agents into consideration when computing the MBB-feasible allocation.
The remaining analysis of this section can be organized into three main steps:
\begin{itemize}
  \setlength\itemsep{0em}
    \item We first show the existence of a coefficient vector $\beta$ with desired properties (carefully designed for achieving weighted proportionality) using the KKM theorem;
    \item Then we establish the structural properties for the MBB-feasibility graph with the fixed coefficient vector $\beta$;
    \item Finally, we identify an MBB-feasible integral allocation, and show that it can be turned into a WPROP outcome with a small amount of total subsidy.
\end{itemize}

\subsection{Fixing the Coefficient Vector}
\label{ssec:fix_beta}

We define a key property called ``$i$-biased'' for the coefficient vectors, and use the KKM theorem to show the existence of a vector that is $i$-biased for every agent $i\in N$.

Recall that the weighted proportional share of agent $i$ is $w_i\cdot v_i(M)$ (which can be positive, negative, or zero).
Our goal is to compute an allocation $\bX$ with a small total subsidy:
\begin{equation*}
    \sum_{i\in N} \left( w_i\cdot v_i(M) - v_i(X_i) \right)^+.
\end{equation*}

For all $\beta\in \Delta_{n-1}$ and allocation $\bX$, we define
\begin{equation*}
    \delta_i(\bX,\beta) = w_i\cdot \beta_i\cdot v_i(M) - p(X_i).
\end{equation*}

Note that for MBB-feasible allocation $\bX$, we have $p(X_i) = \beta_i\cdot v_i(X_i)$, which implies
\begin{equation*}
    \delta_i(\bX,\beta) = \beta_i\cdot (w_i\cdot v_i(M) - v_i(X_i)).
\end{equation*}

Therefore, $\delta_i(\bX,\beta)$ is a scaled measurement for the deficit of agent $i$ for achieving her weighted proportional share.
Therefore, a smaller $\delta_i(\bX,\beta)$ implies a higher utility for agent $i$.
For example, if for some agent $i\in N$ we have $\delta_i(\bX,\beta) > 0$, then to turn the allocation $\bX$ into a WPROP outcome, we need to pay agent $i$ a subsidy of ${\delta_i(\bX,\beta)}/{\beta_i}$.
Hence, in general, we want $\delta_i(\bX,\beta)$'s to be small (e.g., non-positive).
We show that for any MBB-feasible allocation, the sum of $\delta_i(\bX,\beta)$ is always non-positive.

\begin{lemma} \label{lemma:sum_delta_non_positive}
    For all $\beta\in \Delta_{n-1}$ and MBB-feasible allocation $\bX\in \LP^*(\beta)$, we have
    \begin{equation*}
        \sum_{i\in N} \delta_i(\bX,\beta) \leq 0.
    \end{equation*}
\end{lemma}
\begin{proof}
    By dual feasibility of $\bp$, for every agent $i\in N$ we have
    \begin{equation*}
        \beta_i\cdot v_i(M) = \sum_{e\in M} \beta_i\cdot v_i(e)
        \leq \sum_{e\in M} p(e) = p(M).
    \end{equation*}

    Therefore, we have
    \begin{equation*}
        \sum_{i\in N} \delta_i(\bX, \beta)
        = \sum_{i\in N} \left( w_i\cdot \beta_i\cdot v_i(M) - p(X_i) \right)
        \leq \sum_{i\in N} w_i\cdot p(M) - \sum_{i\in N} p(X_i) = 0.
        \qedhere
    \end{equation*}
\end{proof}

An equivalent interpretation of Lemma~\ref{lemma:sum_delta_non_positive} is by weak duality: $\sum_{i\in N} w_i\cdot \beta_i\cdot v_i(M)$ is the objective of the feasible primal solution where $x_i(e) = w_i$ for all $i\in N$ and $e\in M$, while $\sum_{i\in N} p(X_i) = p(M)$ is the objective of the optimal dual solution.

\begin{definition}[$i$-biased]
    A coefficient vector $\beta\in \Delta_{n-1}$ is called $i$-biased for some agent $i\in N$, if there exists an MBB-feasible integral allocation $\bX\in \LP^*(\beta)$ such that
    \begin{equation*}
        \delta_i(\bX,\beta) = \min_{j\in N} \left\{ \delta_j(\bX,\beta) \right\},
    \end{equation*}
    where $\bp$ is the optimal solution to $\mathsf{Dual}(\beta)$.
\end{definition}

Intuitively speaking, a coefficient vector $\beta\in \Delta_{n-1}$ is $i$-biased if there exists an MBB-feasible integral allocation $\bX$ that is in favor of agent $i$, in the sense that its deficit to the weighted proportional share (scaled by a factor of $\beta_i$) is the smallest among all agents.
Note that if the vector $\beta$ is $i$-biased, then for the corresponding allocation $\bX$ we have $\delta_i(\bX,\beta) \leq 0$ (which implies that $\bX$ is WPROP for agent $i$ when $\beta_i>0$), because the sum of $\delta$'s is non-positive, by Lemma~\ref{lemma:sum_delta_non_positive}.
Furthermore, a vector $\beta$ can be $i$-biased for multiple agents $i\in N$.

\smallskip

Let $C_i\subseteq \Delta_{n-1}$ be the collection of coefficient vectors $\beta$ such that $\hat{\beta}$ is $i$-biased:
\begin{equation*}
    C_i = \left\{ \beta\in \Delta_{n-1} : \hat{\beta} \text{ is $i$-biased } \right\}.
\end{equation*}

Notice that $C_i$ is a subset of vectors $\beta$ but the condition for $\beta\in C_i$ is based on $\hat{\beta}$, a shrunken version of $\beta$ satisfying $\hat{\beta}>0$.
In the following, we show that $C_1,C_2,\ldots,C_n$ are closed, and they cover the whole simplex $\Delta_{n-1}$.

\begin{restatable}[Closedness]{lemma}{thmClosednessOfC}
\label{lemma:closedness_of_C}
For all $i \in N$, the subset $C_i \subseteq \Delta_{n-1}$ is closed.
\end{restatable}
\begin{proof}[Proof sketch]
(See \cref{sec:missing_proofs} for the full proof.)
Suppose there exists a sequence $B = (\beta^{(1)}, \beta^{(2)}, \ldots)$ such that each element of $B$ lies in $C_i$.
We must show that if $B$ converges to $\beta$, then $\beta \in C_i$.
Each $\betahat^{(t)}$ has a corresponding MBB-feasible allocation $\bX^{(t)}$,
and since the set of integral allocations is finite,
we can assume \wLoG{} that all these MBB-feasible allocations are identical.
Denote this common allocation by $\bX$.
Using \cref{observation:unique_dual}, we can define a price $p$ for $\betahat$,
and one can show that $\bX$ is MBB-feasible for $p$.
We prove that $i \in \argmin_{j \in N} \{\delta_j(\bX, \betahat)\}$,
which implies $\beta \in C_i$.
\end{proof}

\begin{lemma}[Covering Condition] \label{lemma:beta_i>0_and_beta_in_C_i}
    For sufficiently small $\epsilon > 0$, for all vector $\beta\in \Delta_{n-1}$, there exists an agent $i\in N$ such that $\beta_i > 0$ and $\beta\in C_i$.
\end{lemma}
\begin{proof}
    Fix an arbitrary allocation $\bX \in \LP^*(\hat{\beta})$ and let $i$ be the agent with minimum $\delta_i(\bX,\hat{\beta})$.
    By definition, $\hat{\beta}$ is $i$-biased, and we have $\beta\in C_i$.
    If $\beta > 0$ then we have $\beta_i > 0$ and $\beta\in C_i$, and the lemma follows.
    Otherwise ($\beta_j = 0$ for some agent $j$), we show that we still have $\beta_i > 0$.

    Recall that $\hat{\beta}$ is obtained by shrinking $\beta$ to the center of the simplex by a factor of $(1-n\epsilon)$.
    If $\beta_j = 0$, then we have $\hat{\beta}_j = (1-n\epsilon)\beta_j + \epsilon = \epsilon$, which is very small.
    We show that to maximize the objective of $\LP(\hat{\beta})$, we should not allocate any good to agent $j$.
    Let $k$ be the agent with the maximum $\beta_k$, which must satisfy $\beta_k > 1/n$.
    By the dual feasibility of $\mathsf{Dual}(\hat{\beta})$, for all items $e\in M^+$ (recall that every item in $M^+$ has positive values to all agents), we have
    \begin{equation*}
        p(e) \geq \hat{\beta}_k\cdot v_k(e)
        = \left( (1-n\epsilon)\cdot \beta_k +\epsilon \right)\cdot v_k(e)
        > \frac{1}{n} \cdot v_k(e)
        > \epsilon\cdot v_j(e),
    \end{equation*}
    where the last inequality holds by definition of $\epsilon$.
    Since the inequality is strict, we conclude that $e$ is not an MBB item for agent $j$.
    Therefore, we have $M_j\cap M^+ = \varnothing$, which implies that $p(X_j) \leq 0$, as all chores have negative prices.

    Similarly, for all $e\in M^-$, we have (note that $e$ has negative values to all agents)
    \begin{equation*}
        p(e) \geq \epsilon\cdot v_j(e) > \frac{1}{n} \cdot v_k(e) > \hat{\beta}_k \cdot v_k(e),
    \end{equation*}
    which implies that agent $k$ has no MBB chores, i.e., $M_k\cap M^- = \varnothing$.
    Therefore, we have $p(X_k) \geq 0$, as all goods have positive prices.

    Next, we show that $\delta_j(\bX,\hat{\beta})$ is not the minimum among all agents, which implies $\beta_i > 0$.

    If $v_j(M) > 0$ then we are done, because
    \begin{equation*}
        \delta_j(\bX,\hat{\beta}) = w_j\cdot \epsilon \cdot v_j(M) - p(X_j) > 0,
    \end{equation*}
    but Lemma~\ref{lemma:sum_delta_non_positive} implies that the minimum $\delta_i(\bX,\hat{\beta})$ is non-positive.

    If $v_j(M) < 0$, then we also have $v_k(M) < 0$ by consistency of the instance.
    Since $p(X_j)\leq 0$ and $p(X_k)\geq 0$, we have
    \begin{equation*}
        \delta_j(\bX,\hat{\beta}) \geq w_j\cdot \epsilon\cdot v_j(M) > w_k\cdot\frac{1}{n}\cdot v_k(M) \geq \delta_k(\bX,\hat{\beta}),
    \end{equation*}
    where the second inequality holds by definition of $\epsilon$:
    \begin{equation*}
        \epsilon < \frac{w_k\cdot |v_k(M)|}{nm\cdot \max_{e\in M} \{ |v_j(e)| \}} \leq \frac{w_k\cdot v_k(M)}{w_j\cdot v_j(M)\cdot n}.
    \end{equation*}

    Therefore $\delta_j(\bX,\hat{\beta})$ is not the minimum and we have $\beta_i > 0$, as claimed.
\end{proof}

Then, we apply the KKM fixed point theorem, which is formally stated below.

\begin{theorem}[\cite{knaster1929}] \label{theorem:KKM}
    Let $C_1,C_2,\ldots,C_n$ be closed subsets of $\Delta_{n-1}$, such that for every $\beta\in \Delta_{n-1}$, there exists an index $i\in [n]$ such that $\beta_i > 0$ and $\beta\in C_i$. Then the intersection of all the subsets is non-empty, i.e., $\bigcap_{i\in[n]} C_i \neq \varnothing$.
\end{theorem}

Applying the KKM Fixed Point Theorem to the sets $\{C_1,C_2,\ldots,C_n\}$ we have constructed above, we conclude that there exists $\beta\in \Delta_{n-1}$ such that $\hat{\beta}$ is $i$-biased for every agent $i\in N$.
Since $\hat{\beta} > 0$, we have the following.

\begin{lemma} \label{lemma:beta_in_intersection_of_C}
    There exists $\beta\in \Delta_{n-1}$ with $\beta>0$ that is $i$-biased for all $i\in N$.
\end{lemma}

In what follows, we fix $\beta > 0$ to be this specific coefficient vector and let $\bp$ be the corresponding price vector, i.e., the optimal dual solution.
Since $\beta$ is fixed throughout the remainder of the analysis, in what follows, we use $\delta_i(\bX)$ to denote $\delta_i(\bX,\beta)$.

\subsection{Property of the MBB Graph}
\label{ssec:property_of_MBB}

For the fixed prices $\bp$, we construct a bipartite graph $G(N\cup M, E)$ representing the MBB feasibility.
Let the two sides of the bipartite graph be the set of agents $N$, and the set of items $M$, respectively.
Let there be an edge $(i,e) \in E$ if and only if $p(e) = \beta_i\cdot v_i(e)$.
We show that the graph is a forest, using the non-degeneracy of the instance (see Appendix~\ref{sec:missing_proofs} for a proof).

\begin{lemma} \label{lemma:no_cycle_in_MBB_graph}
    The MBB-feasibility graph for a non-degenerate instance is a forest.
\end{lemma}

Recall that our goal is to compute an MBB-feasible allocation $\bX$ that can be turned into a WPROP outcome with a small amount of total subsidy.
The above lemma implies that the MBB feasibility graph is sparse, which is (intuitively speaking) helpful for identifying the allocation.
In the following, we show that we can make the graph even sparser by removing edges incident to leaf item nodes.

\begin{definition}[Reserved Items]
    We call an item $e\in M$ \emph{reserved} for agent $i\in N$, if $e$ is an MBB item only to agent $i$, i.e., $e$ has only one neighbor in $G(N\cup M,E)$, which is $i$.
    Let $R_i \subseteq M_i$ be the set of reserved items of agent $i$.
\end{definition}

Note that the sets of reserved items $R_1,\ldots,R_n$ are disjoint.
Since there is no freedom to decide the allocation of reserved items (subject to the MBB-feasibility constraint), we can hide all these items in the graph when identifying the MBB-feasible allocation.
Hence, it remains to consider a forest in which all leaf nodes are agents, which implies that there are at most $n-1$ items in the graph.
We call these items \emph{shared items} as each of them has a degree of at least $2$.
Note that for all MBB-feasible allocation $\bX$, for all $i\in N$ we have $R_i\subseteq X_i\subseteq M_i$.
Finally, we establish an important property for MBB feasibility,
making use of the fixed coefficient vector $\beta$.
Indeed, this is the only property we need from $\beta$.

For each agent $i\in N$, we define
\begin{equation*}
    q_i = \max_{X_i:R_i\subseteq X_i\subseteq M_i}\{p(X_i)\}
    = p\left( R_i\cup(M_i\cap M^+) \right)
\end{equation*}
as the maximum possible price of $i$ under MBB-feasible allocations.
Note that $q_i$ can be obtained by allocating all goods in $M_i$ to $X_i$, besides all reserved items.

By Lemma~\ref{lemma:beta_in_intersection_of_C}, for every agent $i\in N$, there exists an MBB-feasible allocation $\bX$ such that $\delta_i(\bX)$ is minimum among all agents, which yields the following lemma.

\begin{lemma}\label{lemma:allocation_with_price_at_least_p(R)-wp(M)}
    For all agent $i\in N$, there exists an allocation $\bX \in \LP^*(\beta)$ in which
    \begin{equation*}
        \min_{j\in N}\left\{ \delta_j(\bX) \right\} = \delta_i(\bX) = w_i\cdot \beta_i\cdot v_i(M) - p(X_i) \geq w_i\cdot \beta_i\cdot v_i(M) - q_i.
    \end{equation*}
    In other words, in $\bX$, we have $\delta_j(\bX) \geq w_i\cdot \beta_i\cdot v_i(M) - q_i$ for all $j\in N$.
\end{lemma}

\subsection{Identifying the Allocation}
\label{ssec:identify_allocation}

Next, we identify an integral MBB-feasible allocation $\bX$ that requires a small amount of subsidy.
For ease of notation, for any MBB-feasible bundle $X_i$ allocated to agent $i$, let
\begin{equation*}
    \hat{p}(X_i) = p(X_i) + \max\left\{\left\{0\right\}\cup\left\{p(e):e\in M_i\setminus X_i\right\}\cup\left\{-p(e):e\in X_i\setminus R_i\right\}\right\}
\end{equation*}
be the maximum price of $X_i$ after adding or removing at most one MBB item (except for those in $R_i$).
Note that if $M_i\setminus X_i$ contains only chores (whose prices are negative) and $X_i\setminus R_i$ contains only goods (whose prices are positive), then we can choose not to alter the bundle, as ensured by including $0$ in the maximum. This definition also applies when either set is empty.
Therefore we always have $\hat{p}(X_i) \geq p(X_i)$.

\smallskip

Let $\bX$ be the MBB-feasible allocation that maximizes
\begin{equation*}
    \Gamma(\bX):= \min_{i\in N}\{ \delta_i(\bX) \} = \min_{i\in N}\{ w_i\cdot \beta_i\cdot v_i(M) - p(X_i) \},
\end{equation*}
and subject to which maximizes the number of agents satisfying the following condition:
\begin{equation}
    w_i\cdot \beta_i\cdot v_i(M) - \hat{p}(X_i) \leq \Gamma(\bX). \label{equation:hatp-w_at_most_Gamma}
\end{equation}

In the following, we fix this allocation, and let $\Gamma = \Gamma(\bX)$, which is the maximum over all MBB-feasible allocations.
Note that we have $\Gamma \leq 0$ by Lemma~\ref{lemma:sum_delta_non_positive}.
Therefore, if Condition~\eqref{equation:hatp-w_at_most_Gamma} is satisfied for agent $i$, then the allocation $\bX$ is WPROP1 for $i$ (see Lemma~\ref{lemma:WPROP1}).

As a corollary of Lemma~\ref{lemma:allocation_with_price_at_least_p(R)-wp(M)}, we obtain the following.

\begin{corollary} \label{corollary:Gamma_at_least_max_p(R)-w}
    We have $\Gamma \geq \max_{i\in N} \{ w_i\cdot \beta_i\cdot v_i(M) - q_i \}$.
\end{corollary}
\begin{proof}
    By Lemma~\ref{lemma:allocation_with_price_at_least_p(R)-wp(M)}, for every agent $i\in N$ there exists an MBB allocation $\bX'$ such that
    \begin{equation*}
        \Gamma(\bX') = \min_{j\in N} \{ \delta_j(\bX') \} \geq w_i\cdot \beta_i\cdot v_i(M) - q_i.
    \end{equation*}
    The corollary follows since $\Gamma = \max_{\bX'}\{ \Gamma(\bX') \}$.
\end{proof}

Corollary~\ref{corollary:Gamma_at_least_max_p(R)-w} implies that under allocation $\bX$, for every agent $i$, if $X_i$ is ``one-item-far'' from the allocation that defines $q_i$, then Condition~\eqref{equation:hatp-w_at_most_Gamma} is satisfied.
Recall that $q_i$ is achieved when agent $i$ receives all her MBB goods and no MBB chores in $M_i\setminus R_i$.
Hence, ``one-item-far'' means that $X_i$ contains all MBB goods and $X_i\setminus R_i$ contains at most one chore; or $X_i$ contains all but one MBB good, and $X_i\setminus R_i$ contains no chore.
Next, we present the most important property on allocation $\bX$.

\begin{lemma} \label{lemma:condition_holds_for_all_agents}
    In allocation $\bX$, all agents satisfy Condition~\eqref{equation:hatp-w_at_most_Gamma}.
\end{lemma}
\begin{proof}
    Assume otherwise, i.e., there exists an agent $i$ for which $w_i\cdot \beta_i\cdot v_i(M) - \hat{p}(X_i) > \Gamma$.
    We can alter the allocation, preserving $\delta_j(\bX) \geq \Gamma$ for all agents $j\in N$, but increasing the number of agents satisfying Condition~\eqref{equation:hatp-w_at_most_Gamma}, which leads to a contradiction.

    Recall that $G(N\cup M, E)$ is a forest.
    Thus, there exists a tree containing $i$.
    Root the tree at $i$.
    The child items of $i$ are $M_i\setminus R_i$.
    Also, recall that all leaf nodes are agents.
    We repeat the following operations to increase $p(X_i)$, until Condition~\eqref{equation:hatp-w_at_most_Gamma}, i.e., $w_i\cdot \beta_i\cdot v_i(M) - \hat{p}(X_i) \leq \Gamma$ is satisfied for the first time:
    \begin{itemize}
        \item If there exists some chore in $X_i\setminus R_i$, remove it;
        \item otherwise, if there exists some good in $M_i\setminus X_i$, we include it into $X_i$.
    \end{itemize}

    Note that
    \begin{itemize}
        \item The condition is guaranteed to be satisfied at some moment, because when all the child chores are removed and all child goods are added, we have $p(X_i) = q_i$.
        \item Immediately before each transfer, failure of the condition gives $\delta_i\left(\bX\right)-\left(\hat{p}\left(X_i\right)-p\left(X_i\right)\right)>\Gamma$. The increase in $p\left(X_i\right)$ caused by the transferred item is at most $\hat{p}\left(X_i\right)-p\left(X_i\right)$, evaluated before the transfer. Thus, after the transfer, we still have $\delta_i\left(\bX\right)>\Gamma$.
    \end{itemize}

    Note that each good $e$ added to $X_i$ during the process is taken from some agent $j$ that is a child of item $e$ in the tree.
    For each released chore $e\in X_i\setminus R_i$, we assign it to an arbitrary agent $j$ sharing $e$ with $i$. Again, $j$ is a child of item $e$ in the tree.
    For each such agent $j$, the reallocation decreases $p(X_j)$ (since either a good is taken, or an extra chore is assigned), which maintains $\delta_j(\bX) \geq \Gamma$ but could possibly lead to Condition~\eqref{equation:hatp-w_at_most_Gamma} being violated for agent $j$.

    If the condition is not violated, then we are done with agent $j$.

    Otherwise, we recurse on the subtree rooted at agent $j$, i.e., by repeatedly removing chores from $X_j\setminus R_j$, or adding goods to $X_j$ until Condition~\eqref{equation:hatp-w_at_most_Gamma} is satisfied.
    Note that the reallocation only applies to child items of $j$: the parent of $j$ (which is an item) will not be reallocated. The same before-transfer inequality shows that every improving transfer at $j$ leaves $\delta_j\left(\bX\right)>\Gamma$, while the agent harmed by the transfer has an increased deficit. Since the graph is a tree, descendant transfers change neither an ancestor's bundle nor its one-item price improvement. The fixed parent item is the only item that can prevent $j$ from attaining $q_j$.
    \begin{itemize}
        \item If $j$ is a leaf node (for which we have $|M_j\setminus R_j|\leq 1$), then Condition~\eqref{equation:hatp-w_at_most_Gamma} must be satisfied.
        \item Similarly, if $j$ is an internal node, then by removing all its child chores and adding all its child goods, Condition~\eqref{equation:hatp-w_at_most_Gamma} must be satisfied, because when this happens $X_j$ is one-item-far from the allocation that defines $q_j$.
    \end{itemize}

    The recursion terminates because it moves only down the finite tree, and each agent transfers each child item at most once. Every previously satisfying agent remains satisfying or is restored to satisfying, and the initially chosen agent $i$ becomes satisfying; untouched agents retain their previous status. All deficits remain at least $\Gamma$, so maximality of $\Gamma$ implies that the minimum deficit is still $\Gamma$. This strictly improves the tie-break, a contradiction.
\end{proof}

As a corollary, we show that the allocation is WPROP1.

\begin{lemma} \label{lemma:WPROP1}
    The allocation $\bX$ is WPROP1.
\end{lemma}
\begin{proof}
    Fix any agent $i\in N$.
    Since agent $i$ satisfies Condition~\eqref{equation:hatp-w_at_most_Gamma}, we have
    \begin{equation*}
        w_i\cdot \beta_i\cdot v_i(M) - \hat{p}(X_i) \leq \Gamma \leq 0.
    \end{equation*}

    By definition, we have either $\hat{p}(X_i) = p(X_i)$, or $\hat{p}(X_i) = p(X'_i)$, where $X'_i$ can be obtained by adding an MBB item to or removing an MBB item from $X_i$.
    Furthermore, for MBB items, we have $p(X_i) = \beta_i\cdot v_i(X_i)$ (resp. $p(X'_i) = \beta_i\cdot v_i(X'_i)$).
    Then by dividing by $\beta_i$ (which is positive) on both sides of the above inequality, we obtain that either $v_i(X_i) \geq w_i\cdot v_i(M)$, or $v_i(X'_i) \geq w_i\cdot v_i(M)$, which implies that $\bX$ is WPROP1 for agent $i$.
\end{proof}

\subsection{The Total Subsidy}
\label{ssec:total_subsidy}
Finally, we bound the subsidy required to turn $\bX$ into WPROP.
Since $\bX$ is also fixed in the remainder of the analysis, we use $\delta_i$ to denote $\delta_i(\bX)$.
Recall that in $\bX$, for every agent $i\in N$ we have
\begin{equation}
    \delta_i = w_i\cdot \beta_i\cdot v_i(M) - p(X_i) \geq \Gamma,
    \label{equation:all_delta_>=Gamma}
\end{equation}
and
\begin{equation}
    w_i\cdot \beta_i\cdot v_i(M) - \hat{p}(X_i) \leq \Gamma.
    \label{equation:all_delta_hat_<=Gamma}
\end{equation}

Recall that $\delta_i = \beta_i\cdot (w_i\cdot v_i(M) - v_i(X_i))$ represents the deficit (scaled by $\beta_i$) of agent $i$ towards achieving her weighted proportional share, and we need to pay $(\delta_i/\beta_i)^+$ amount of subsidy to agent $i$ to obtain a WPROP outcome.
Equations~\eqref{equation:all_delta_>=Gamma} and~\eqref{equation:all_delta_hat_<=Gamma} combined state that the difference in $\delta$ between any two agents is upper bounded by the absolute price of one item.
On the other hand, Lemma~\ref{lemma:sum_delta_non_positive} states that the sum of $\delta$ is non-positive.
Therefore, intuitively speaking, $\sum_{i\in N} (\delta_i/\beta_i)^+$, which is the sum of positive $\delta_i/\beta_i$'s, cannot be too large.

\paragraph{Comparison with Existing Work.}
In some sense, our analysis for upper bounding the total subsidy is similar to existing works, e.g., see Lemma 3.3 of~\citet{ijcai/WuXZ25}, where they show that if $\sum_{i\in N} d_i \leq 0$, where $d_i = w_i\cdot v_i(M) - v_i(X_i)$, and the difference between any $d_i$ and $d_j$ is upper bounded by $1$, then the total subsidy is at most $\tau(n)$.
However, our analysis is also fundamentally different in that we only have $\sum_{i\in N} \delta_i \leq 0$, which does not necessarily imply $\sum_{i\in N} d_i \leq 0$ (see Example~\ref{example:positive_total_surplus} for an instance).
This is because the allocation we compute is optimal to $\LP(\beta)$, where in the objective each agent $i$ is scaled by a factor of $\beta_i$.
To resolve this problem, we formulate the subsidy bound as a concave function and apply Jensen's inequality to bound the total subsidy, which could be of independent interest for further research on fair allocation with subsidy.

\begin{example} \label{example:positive_total_surplus}
    Consider the following example with a single item $e$ and two equal-weight agents.\footnote{No choice of subsidies makes this particular allocation envy-free. Thus, MBB-feasibility for an arbitrary positive coefficient vector alone does not imply envy-freeability.}

    \newcommand{\circval}[1]{\tikz[baseline=-0.6ex]{\node[draw,circle,inner sep=1pt] {$#1$};}}
    \renewcommand{\arraystretch}{1.1}
    \setlength{\tabcolsep}{8pt}
    \begin{center}
    \begin{tabular}{c|c|c|c}
        &  $v_i(e)$ & $w_i$ & $\beta_i$  \\ \hline
        Agent 1   & $1$ & $0.5$ & $0.1$   \\
        Agent 2   & \circval{0.5} & $0.5$ & $0.9$
    \end{tabular}
    \end{center}

    The optimal solution to $\LP(\beta)$ allocates item $e$ to agent $2$ and we have
    \begin{align*}
        & \delta_1 = \beta_1\cdot (w_1\cdot v_1(e) - v_1(X_1)) = 0.05,\\
        & \delta_2 = \beta_2\cdot (w_2\cdot v_2(e) - v_2(X_2)) = -0.225,\\
        & \delta_1 + \delta_2 < 0,
    \end{align*}
    which is consistent with Lemma~\ref{lemma:sum_delta_non_positive}.
    However, if we measure the agents' distance to proportionality under this allocation, we get $d_1 = 0.5$ and $d_2 = -0.25$, whose sum is positive.
\end{example}

\begin{lemma}
\label{lemma:subsidy_bound}
For the fixed allocation $\bX$ and for canonical instance $\cI$, we have
\begin{equation*}
    \subsidy(\bX, \cI) \le \tau(|N|)\cdot\range(\cI),
\end{equation*}
where $\range(\cI) \defeq \max_{i \in N, e \in M} |v_i(e)|$ and
$\tau(n) \defeq \begin{cases}
        n/4, & \text{when $n$ is even} \\
        (n^2-1)/(4n), & \text{when $n$ is odd}.
    \end{cases}$.
\end{lemma}
\begin{proof}
    For all agent $i\in N$, let
    \begin{equation*}
        s_i = \left( w_i\cdot v_i(M) - v_i(X_i) \right)^+ = \left( \frac{\delta_i}{\beta_i} \right)^+.
    \end{equation*}

    By definition, $(\bX, \bs)$ is a WPROP outcome.
    Next, we upper bound the total subsidy.

    Recall from Lemma~\ref{lemma:sum_delta_non_positive} that $\sum_{i\in N} \delta_i \leq 0$.
    If every $\delta_i\leq0$, no subsidy is needed and the lemma follows. Otherwise, there is a positive deficit and, since their sum is nonpositive, a negative deficit.
    For ease of notation, we reindex the agents and assume
    \begin{equation*}
        \delta_1 \geq \delta_2 \geq \cdots \geq \delta_k > 0
        \geq \delta_{k+1} \geq \cdots \geq \delta_n.
    \end{equation*}

    Then we have (recall that $\beta > 0$)
    \begin{equation*}
        \sum_{i\in N} s_i = \sum_{i\in N} \left( \frac{\delta_i}{\beta_i} \right)^+
        = \sum_{i\leq k} \frac{\delta_i}{\beta_i}.
    \end{equation*}

    Thus $1\leq k<n$. Let $h = -\delta_n > 0$.
    Since
    \begin{equation*}
        \sum_{i\leq k} \delta_i + \sum_{i>k} \delta_i \leq 0,
    \end{equation*}
    we have
    \begin{equation}
        \sum_{i\leq k} \delta_i \leq -\sum_{i>k} \delta_i \leq -(n-k)\cdot \delta_n = (n-k)\cdot h. \label{equation:delta_<=k_upper_bound}
    \end{equation}

    By Inequalities~\eqref{equation:all_delta_>=Gamma} and~\eqref{equation:all_delta_hat_<=Gamma}, for all $i\leq k$, we have
    \begin{align*}
        -h = \delta_n \geq \Gamma &\geq w_i\cdot \beta_i\cdot v_i(M) - \hat{p}(X_i) \\
        &= \delta_i-\left(\hat{p}(X_i)-p(X_i)\right)\geq\delta_i-\beta_i\cdot\range(\cI),
    \end{align*}
    where the last inequality uses the bound $\beta_i\cdot\range(\cI)$ on each one-item price improvement.

    Therefore we have $\beta_i \geq \frac{\delta_i + h}{\range(\cI)}$, which implies
    \begin{equation*}
        \sum_{i\in N} s_i
        = \sum_{i\leq k} \frac{\delta_i}{\beta_i}
        \leq \range(\cI)\cdot \sum_{i\leq k} \frac{\delta_i}{\delta_i + h}.
    \end{equation*}

    Let $f(x) = \frac{x}{x+h}$ be a function defined on $(0,+\infty)$.
    It is easy to verify that
    \begin{equation*}
        f'(x) = \frac{h}{(x+h)^2} > 0,
        \qquad \text{and} \qquad
        f''(x) = \frac{-2h}{(x+h)^3} < 0.
    \end{equation*}

    Therefore, $f(x)$ is a concave increasing function.
    Then by Jensen's inequality, we obtain
    \begin{align*}
        \frac{\sum_{i\in N} s_i}{\range(\cI)}
        & \leq \sum_{i\leq k} f(\delta_i) \leq k\cdot f\left( \frac{\sum_{i=1}^k \delta_i}{k} \right)
        \leq k\cdot f\left( \frac{(n-k)h}{k} \right) \\
        & = \frac{k\cdot (n-k)\cdot h}{(n-k)\cdot h + h\cdot k} = \frac{k(n-k)}{n} \leq
        \begin{cases}
            \frac{n}{4}, & \text{when $n$ is even} \\
            \frac{n^2-1}{4n}, & \text{when $n$ is odd}
        \end{cases}
    \end{align*}
    where the third inequality follows from Equation~\eqref{equation:delta_<=k_upper_bound}.
    Rearranging yields the lemma.
\end{proof}

Summarizing the above, we obtain the following main results of the paper.

\begin{theorem} \label{thm:canonical-existence}
For any canonical mixed manna instance $\Ical$ with $n$ agents with additive valuation functions and general weights,
there always exists an fPO and WPROP1 allocation that can be made WPROP with total subsidy at most $\tau(n)\cdot\range(\Ical)$.
\end{theorem}
\begin{proof}
In \cref{ssec:fix_beta}, we fixed a vector $\beta \in \Delta_{n-1}$ such that $\beta > 0$.
In \cref{ssec:identify_allocation}, we fixed an allocation $\bX$ that is MBB-feasible for $\beta$.
By \cref{lemma:MBB-feasible_implies_PO_when_beta>0}, $\bX$ is fPO, and by \cref{lemma:WPROP1}, it is WPROP1.
By \cref{lemma:subsidy_bound}, $\subsidy(\bX, \Ical) \le \tau(|N|)\cdot\range(\Ical)$.
\end{proof}

\begin{theorem} \label{thm:existence}
For any mixed manna instance $\Ical$ with $n$ agents with additive valuation functions and general weights,
there always exists a WPROP1 allocation that can be made WPROP with total subsidy at most $\tau(n)\cdot\range(\Ical)$.
\end{theorem}
\begin{proof}
Let $\Ical = (N, M, v, w)$. Handle the one-agent and zero-range cases as in Section~\ref{sec:prelims}.
Let $\Icalhat \defeq (N', M', \vhat, w')$ be the instance
obtained after applying the first three steps of canonicalization to $\Ical$ (see \cref{sec:redn-canon}). If $M'=\varnothing$, Step 1 already gives the required allocation.
\WLoG{}, let $\range(\Ical) = 1$. Then $\range(\Icalhat)\leq1$.

Let $\mathcal{P}$ be the set of all disjoint $|N'|$-bundles of $M'$, and let
\begin{align*}
\lambda_0 &\defeq \frac{1}{2}\cdot\min\left\{\begin{array}{ll}
        \displaystyle \min\left\{ \left|v_i(S) - w_i\cdot v_i(M)\right|:
            i\in N, \; S \subseteq M,\; v_i(S) \neq w_i \cdot v_i(M) \right\},
        \\ \displaystyle \min\left\{ \sum_{i\in N'} \vhat_i(A_i) - \sum_{i\in N'} \vhat_i(B_i):
            A, B \in \mathcal{P}, \sum_{i\in N'} \vhat_i(A_i) > \sum_{i\in N'} \vhat_i(B_i) \right\}
    \end{array}\right\},
\\ \eta_0 &\defeq \frac{\lambda_0}{m}.
\end{align*}

For all $t \in \N$, let $\Ical^{(t)}$ be the instance obtained by randomly perturbing $\Icalhat$,
with the parameters $\lambda$ and $\eta$ set to $\lambda_t \defeq \lambda_0/(t+1)$
and $\eta_t \defeq \eta_0/(t+1)$, respectively.

For each $t$, fix a canonical realization, which exists by Lemma~\ref{lemma:constructed_instance_is_canonical}. Then instance $\Ical^{(t)}$ is canonical and $\range(\Ical^{(t)}) \le 1 + \eta_t$.
Thus, by \cref{thm:canonical-existence}, there is a WPROP1 allocation $\bX^{(t)}$ such that
\begin{equation*}
    \subsidy(\bX^{(t)}, \Ical^{(t)}) \le \tau(|N'|)\cdot \range(\Ical^{(t)}) \le \tau(n)\cdot (1+\eta_t).
\end{equation*}

Since the number of integral allocations is finite, there exists an allocation $\bX'$
that appears infinitely often in the sequence $(\bX^{(t)})_{t \in \N}$.
Let $T \defeq \{t \in \N: \bX^{(t)} = \bX'\}$ be the indices of $\bX'$.
By \cref{lemma:reduction_to_canonical_instances}, there exists an allocation $\bX$ for $\Ical$
that is WPROP1 and such that for all $t \in T$, we have
\begin{equation*}
    \subsidy(\bX, \Ical) \le \subsidy(\bX', \Ical^{(t)}) + \lambda_t \le \tau(n)\cdot (1 + \eta_t) + \lambda_t.
\end{equation*}

Thus, we have
\begin{equation*}
    \subsidy(\bX, \Ical) \le \inf_{t \in T} \left\{ \tau(n)\cdot (1 + \eta_t) + \lambda_t \right\} = \tau(n).\qedhere
\end{equation*}
\end{proof}

\section{Polynomial-time Algorithm for Few Types of Agents}
\label{sec:few-types}

We now consider the setting where the $n$ agents belong to $k$ different groups,
and the agents in each group have the same valuation function.
When $k$ is a constant, we design a polynomial-time algorithm
to achieve a subsidy of $\tau(k) + \tau(n)+\varepsilon$ for any prescribed $\varepsilon>0$. The running time is polynomial in the binary input length and $\log\left(1+1/\varepsilon\right)$.
If the compressed instance is already canonical, the $\varepsilon$ term is unnecessary, and we also get Pareto optimality.
For all computational statements, valuations and weights are rational numbers encoded in binary. We normalize to $\range\left(\Ical\right)\leq1$; subsidy bounds scale by the range in the original units.
We begin by looking at two special cases of the problem: $k = 1$ and $k = n$.

When $k = 1$, i.e., all agents have the same valuation function,
we give a simple greedy algorithm that achieves a subsidy of $\tau(n)$.
When $k = n$, we first canonicalize the instance (as in \cref{sec:redn-canon})
and then we show how to iterate over all integral fPO allocations in polynomial time.
Since there always exists an fPO allocation with subsidy at most $\tau(n)$
(by \cref{thm:canonical-existence}),
we can just compute the total subsidy required by each of these fPO allocations
and output the allocation with the least subsidy after mapping back to the original instance. A rational-separation argument in Section~\ref{sec:fpo-enum} makes the bound exact.

Then we show how to use these special cases to solve the general case.
We give a two-stage allocation algorithm, which first allocates the items among the $k$ groups,
and then further subdivides the items among the group members.
The first stage uses the fPO-enumeration algorithm, and the second stage uses the greedy algorithm for identical valuations.
We then prove that the total subsidy can be upper-bounded by the sum of the subsidies required in each stage.

\subsection{Greedy Algorithm for Identical Valuations}
\label{sec:greedy}

Let $\Ical \defeq (N, M, v, w)$ be a fair division instance,
where $v_i = u$ for all $i \in [n]$, i.e.,
$u(e)$ is the value every agent has for item $e \in M$.
Assume $u(e) \in [-1, 1]$ for all $e \in M$ \wLoG.

We first convert the instance (that contains a mixture of goods and chores) to one having only goods or only chores.
If $M=\varnothing$, return the empty allocation. Otherwise, we repeatedly take distinct items $g$ and $c$ such that $u(g)\geq0$ and $u(c)\leq0$,
and then combine them into a single item having value $u(g) + u(c)$.
When no more operations of this form are possible, we are left with
a set $M$ of multiple composite items such that exactly one of the following is true:
\begin{enumerate}
\item (only goods) $u(e) > 0$ for all $e \in M$.
\item (only chores) $u(e) < 0$ for all $e \in M$.
\item (single neutral item) $|M| = 1$ and $u(M) = 0$.
\end{enumerate}

The first two cases correspond to goods-only and chores-only, respectively.
Moreover, the proportional share of each agent remains the same,
and we continue to have $u(e) \in [-1, 1]$ after this preprocessing step.
In the third case, we can just allocate $M$ to an arbitrary agent, and we require a subsidy of 0.
Thus, if we can find an allocation $\bX$ with total subsidy $s$ in the preprocessed instance,
then $\bX$ would require a total subsidy of $s$ for the original instance too.

Now that we have reduced the problem to that of only goods or only chores,
we can use techniques from \citet{wine/WuZZ23} to get a subsidy of at most $\tau(n)$.
The key insight for goods is that we can allocate more than $|M| - n$ goods among the agents
such that no one gets more than their proportional share,
and this effectively reduces the problem to the simple case of fewer than $n$ goods.
A similar idea works for chores, too.
The algorithm runs in $O((m+n)\log(m+n))$ time.
See \cref{sec:greedy-details} for the full details.

\subsection{Enumerating fPO Allocations}
\label{sec:fpo-enum}

In this section, we consider the case when agents have general additive valuations
and prove that any non-degenerate strictly-objective fair division instance
having $n$ agents and $m$ items has at most $(m+1)^{\binom{n}{2}}$ (integral) fPO allocations,
and we give an $O\left(\left(m+1\right)^{\binom{n}{2}}\cdot\poly\left(m,n\right)\right)$-time algorithm to enumerate them, suppressing polynomial dependence on the input bit length. A one-agent instance is handled by assigning all items to that agent.
Our approach is based on a similar algorithm for chores by \citet{branzei2024algorithms} and for mixed manna by \cite{corr/abs-2507-03946}.

To build intuition, we first consider the setting of two agents.
If we want to distribute goods $M$ among two agents $i$ and $j$ \emph{efficiently},
an intuitive idea is to sort the goods in decreasing order of $v_{i}(\cdot)/v_{j}(\cdot)$,
and allocate a prefix of this sequence to $i$ and the remaining goods to $j$.
This is the basis of the well-known \emph{adjusted-winner} procedure \cite{brams1996fair}.
We extend this idea to the mixture of goods and chores.
We first sort the items in decreasing order of $|v_{i}(\cdot)|/|v_{j}(\cdot)|$.
Then, we show that any fPO allocation has the following structure:
there is a prefix $L$ of the sequence of items such that
agent $i$ receives all goods in $L$ and all chores outside $L$,
and agent $j$ receives all chores in $L$ and all goods outside $L$.
We defer the formal proof to \cref{sec:fpo-enum-details:n2}.
Thus, any canonical instance with two agents and $m$ items has at most $m+1$ fPO integral allocations,
all of which can be obtained by sorting the items and iterating over all prefixes.
In fact, every such prefix, including the empty prefix, gives an fPO allocation: choose a positive coefficient ratio $\beta_j/\beta_i$ strictly between the adjacent value ratios, or outside all ratios for an endpoint prefix. The resulting allocation maximizes a strictly positive weighted sum of utilities. For more than two agents, we retain the LP test for fPO (see \cref{thm:check-fpo-using-lp}).
We can extend the above idea to more than two agents.

For any fPO allocation $A$ and any pair $(i, j)$ of agents,
the allocation of $A_{i} \cup A_{j}$ among the two agents is also fPO,
and thus, can be obtained using the above adjusted-winner-like procedure.
Thus, we iterate over all prefixes among every pair of agents,
and define an agent's bundle to be the items she receives in \emph{every head-to-head fight}.
Note that the resulting allocation might be partial: if some items are unallocated, then we discard this allocation.
One can show that every fPO allocation can be obtained this way.
We defer the formal proof to \cref{sec:fpo-enum-details:n2+}.
Thus, there are at most $(m+1)^{\binom{n}{2}}$ fPO allocations,
and we can enumerate a superset of them in $O\left(\left(m+1\right)^{\binom{n}{2}}\cdot n^2\cdot m\cdot\log\left(m+1\right)\right)$ time.
Hence, we can enumerate all fPO allocations in $O\left(\left(m+1\right)^{\binom{n}{2}}\cdot\poly\left(m,n\right)\right)$ time.
By \cref{thm:canonical-existence}, for any canonical instance,
there is an fPO allocation whose total subsidy requirement is at most $\tau(n)$.

\paragraph{Exact computation for fixed $n$.}
The limiting argument in \cref{thm:existence} can be made finite for rational inputs. After normalization, let $Q$ be the product of the denominators of all valuations and weights in the original instance. Its binary encoding length is polynomial in the input length. Every nonzero absolute proportionality gap is at least $1/Q^2$, and every nonzero difference of partial-allocation welfare sums after Steps 1--3 is at least $1/Q$. Thus the smallness conditions in Section~\ref{sec:redn-canon} can be met without enumerating subsets or partial allocations.
Every original total subsidy is an integer multiple of $1/Q^2$. Since the denominator of $\tau\left(n\right)$ divides $4\cdot n$, any original subsidy strictly above $\tau\left(n\right)$ exceeds it by at least $1/\left(4\cdot n\cdot Q^2\right)$. Choose
\begin{equation*}
    \lambda=\frac{1}{8\cdot n\cdot Q^2\cdot\left(1+\tau\left(n\right)/m\right)},\qquad\eta=\frac{\lambda}{m}.
\end{equation*}
These parameters satisfy the corrected perturbation conditions and have polynomial encoding length. Enumerate the fPO allocations of the deterministically perturbed canonical instance and map them back using Lemma~\ref{lemma:reduction_to_canonical_instances}. At least one resulting allocation has original subsidy at most
\begin{equation*}
    \tau\left(n\right)\cdot\left(1+\eta\right)+\lambda
    <\tau\left(n\right)+\frac{1}{4\cdot n\cdot Q^2}.
\end{equation*}
The separation bound forces its subsidy to be at most $\tau\left(n\right)$. To retain WPROP1, discard mapped allocations that do not satisfy it; a WPROP1 candidate with the same subsidy bound exists by \cref{thm:canonical-existence,lemma:reduction_to_canonical_instances}, and WPROP1 can be checked in polynomial time. Selecting the remaining allocation with minimum original subsidy therefore proves the exact polynomial-time guarantee for fixed $n$.

\subsection{Two-Stage Allocation}
\label{sec:two-stage}

We now describe our two-stage allocation process.
Let $\Ical = (N, M, v, w)$ be a fair allocation instance.
Suppose there are $k$ groups of agents, where agents in each group have the same valuation function (and possibly different weights).
Let $N_j \subseteq N$ be agents of group $j \in [k]$,
and $u_j$ be their common valuation function.
Let $W_j \defeq \sum_{i \in N_j} w_i$ be the total weight of agents in group $j$.

Let $\Icalhat \defeq ([k], M, (u_j)_{j=1}^k, (W_j)_{j=1}^k)$ be the
\emph{compressed} instance corresponding to $\Ical$, where we regard agents of the same type as one compressed agent.
The first stage of our algorithm computes a low-subsidy allocation
$\Ahat = (\Ahat_1, \ldots, \Ahat_k)$ for the instance $\Icalhat$.
In the next stage, we create an instance for each group. Specifically, for $j \in [k]$,
let $\Ical^{(j)} \defeq (N_j, \Ahat_j, (v_i)_{i \in N_j}, (w_i/W_j)_{i \in N_j})$.
We then compute a low-subsidy allocation $A^{(j)}$ for each instance $\Ical^{(j)}$.
Finally, we construct an allocation $A$ for the original instance $\Ical$:
$A_i \defeq A^{(j)}_i$ for each $i \in N_j$ and $j \in [k]$.
We can upper bound the subsidy of $A$ in terms of
subsidies of $\Ahat$, $A^{(1)}$, $\ldots$, $A^{(k)}$.
We defer the proof to \cref{sec:two-stage-extra}.

\begin{restatable}{lemma}{thmHierSubsidy}
\label{thm:hier-subsidy}
We have:
$\subsidy(A, \Ical) \le \subsidy(\Ahat, \Icalhat) + \sum_{j=1}^k \subsidy(A^{(j)}, \Ical^{(j)})$.
\end{restatable}

For stage 1, if we combine the canonicalization from \cref{sec:redn-canon}
with the fPO enumeration algorithm of \cref{sec:fpo-enum},
we get a subsidy of $\tau(k)(1+\eta) + \lambda$ in stage 1,
where $\lambda$ and $\eta$ can be made arbitrarily small
(at the expense of increased running time due to
increased bit-complexity of perturbed valuations).
If we use the greedy algorithm of \cref{sec:greedy} in stage 2,
we get a total subsidy of at most $\tau(n)$.
More precisely, stage 2 requires at most $\sum_{j=1}^k\tau\left(\left|N_j\right|\right)\leq\tau\left(n\right)$: for even $n$ this follows from $\tau\left(t\right)\leq t/4$, while for odd $n$ at least one group has odd size and its correction $1/\left(4\cdot\left|N_j\right|\right)$ is at least $1/\left(4\cdot n\right)$.
For any prescribed small constant $\varepsilon>0$, choose the stage-1 parameters so that $\tau\left(k\right)\cdot\eta+\lambda\leq\varepsilon$, in addition to the canonicalization conditions. Such rational parameters have encoding length polynomial in the input length and $\log\left(1/\varepsilon\right)$, by the same denominator bounds used in Section~\ref{sec:fpo-enum}, applied to the compressed instance. Thus,
\begin{equation*}
    \subsidy\left(A,\Ical\right)\leq\tau\left(k\right)+\sum_{j=1}^k\tau\left(\left|N_j\right|\right)+\varepsilon\leq\tau\left(k\right)+\tau\left(n\right)+\varepsilon.
\end{equation*}
The running time is polynomial in the input length and $\log\left(1/\varepsilon\right)$ for fixed $k$. If the compressed instance is already canonical, no perturbation is needed and the same bound holds without $\varepsilon$. In that case the final allocation is also PO: the positive coefficients supporting the stage-1 allocation can be assigned to every member of the corresponding group, and subdividing among agents with identical valuations preserves optimality for that weighted welfare objective.

\section{Conclusion}
In this work, we show that a total subsidy of $\tau(n)$ is sufficient to guarantee the existence of WPROP allocations for indivisible mixed manna, which matches the lower bound established by~\citet{wine/WuZZ23}.
Our results leave several natural questions open.
An immediate direction is whether such allocations can be computed efficiently for general $n$.
Another worthwhile avenue is determining tight subsidy bounds for other fairness notions, such as weighted envy-freeness and the maximin share fairness.
More broadly, our work highlights the potential of the KKM-theorem-based framework in studying share-based fairness. We hope to see this approach applied to additional fairness concepts in future research.

\section*{Declaration for the Use of AI}
The first draft of this paper (with the same set of main results) was finished in January 2026.
The main algorithmic framework and analysis, including the canonical reductions, the use of KKM, and the definition of ``$i$-biased'', were discovered fully by humans.
The manuscript was also written entirely by the authors (other than minor proofreading by AI), who take full responsibility for its content.


\phantomsection
\addcontentsline{toc}{section}{References}
\bibliography{wprop}

@misc{teh2026weightedfairdivisionindivisible,
      title={Weighted Fair Division of Indivisible Mixed Manna}, 
      author={Nicholas Teh},
      year={2026},
      eprint={2609.01580},
      archivePrefix={arXiv},
      primaryClass={cs.GT},
      url={https://arxiv.org/abs/2609.01580}, 
}

@misc{lin2026envyfreenessadditivemixedmanna,
      title={Almost Envy-Freeness for Additive Mixed Manna with Entitlements: Deterministic and Randomized Guarantees}, 
      author={Zehan Lin and Shengxin Liu and Biaoshuai Tao and Shengwei Zhou},
      year={2026},
      eprint={2609.02724},
      archivePrefix={arXiv},
      primaryClass={cs.GT},
      url={https://arxiv.org/abs/2609.02724}, 
}

@article{orl/ChooLSTZ24,
  author       = {Davin Choo and
                  Yan Hao Ling and
                  Warut Suksompong and
                  Nicholas Teh and
                  Jian Zhang},
  title        = {Envy-free house allocation with minimum subsidy},
  journal      = {Oper. Res. Lett.},
  volume       = {54},
  pages        = {107103},
  year         = {2024}
}

@article{IgarshiM26,
  author       = {Ayumi Igarashi and             Fr{\'{e}}d{\'{e}}ric Meunier},
  title        = {Fair and efficient allocation of indivisible items under category constraints},
  journal      = {Games and Economic Behavior},
  year         = {2026}
}

@inproceedings{aaai/SpringerHY24,
  author       = {Max Springer and
                  MohammadTaghi Hajiaghayi and
                  Hadi Yami},
  title        = {Almost Envy-Free Allocations of Indivisible Goods or Chores with Entitlements},
  booktitle    = {{AAAI}},
  pages        = {9901--9908},
  publisher    = {{AAAI} Press},
  year         = {2024}
}

@inproceedings{ijcai/WuX025,
  author       = {Xiaowei Wu and
                  Quan Xue and
                  Shengwei Zhou},
  title        = {A Little Subsidy Ensures {MMS} Allocation for Three Agents},
  booktitle    = {{IJCAI}},
  pages        = {4073--4081},
  publisher    = {ijcai.org},
  year         = {2025}
}

@article{corr/abs-2510-13633,
  author       = {Pooja Kulkarni and
                  Ruta Mehta and
                  Vishnu V. Narayan and
                  Tomasz Ponitka},
  title        = {Online Fair Division With Subsidy: When Do Envy-Free Allocations Exist,
                  and at What Cost?},
  journal      = {CoRR},
  volume       = {abs/2510.13633},
  year         = {2025}
}

@inproceedings{wine/CaragiannisI21,
  author       = {Ioannis Caragiannis and
                  Stavros Ioannidis},
  title        = {Computing Envy-Freeable Allocations with Limited Subsidies},
  booktitle    = {{WINE}},
  series       = {Lecture Notes in Computer Science},
  volume       = {13112},
  pages        = {522--539},
  publisher    = {Springer},
  year         = {2021}
}

@article{ai/KawaseMSTY25,
  author       = {Yasushi Kawase and
                  Kazuhisa Makino and
                  Hanna Sumita and
                  Akihisa Tamura and
                  Makoto Yokoo},
  title        = {Towards optimal subsidy bounds for envy-freeable allocations},
  journal      = {Artif. Intell.},
  volume       = {348},
  pages        = {104406},
  year         = {2025}
}

@inproceedings{ijcai/BarmanKNS22,
  author       = {Siddharth Barman and
                  Anand Krishna and
                  Yadati Narahari and
                  Soumyarup Sadhukhan},
  title        = {Achieving Envy-Freeness with Limited Subsidies under Dichotomous Valuations},
  booktitle    = {{IJCAI}},
  pages        = {60--66},
  publisher    = {ijcai.org},
  year         = {2022}
}

@inproceedings{aaai/000121,
  author       = {Haris Aziz},
  title        = {Achieving Envy-freeness and Equitability with Monetary Transfers},
  booktitle    = {{AAAI}},
  pages        = {5102--5109},
  publisher    = {{AAAI} Press},
  year         = {2021}
}

@article{corr/abs-2505-23251,
  author       = {Yuanyuan Wang and
                  Tianze Wei},
  title        = {Achieving Equitability with Subsidy},
  journal      = {CoRR},
  volume       = {abs/2505.23251},
  year         = {2025}
}

@article{teco/CaragiannisKMPS19,
  author    = {Ioannis Caragiannis and
               David Kurokawa and
               Herv{\'{e}} Moulin and
               Ariel D. Procaccia and
               Nisarg Shah and
               Junxing Wang},
  title     = {The Unreasonable Fairness of Maximum Nash Welfare},
  journal   = {{ACM} Trans. Economics and Comput.},
  volume    = {7},
  number    = {3},
  pages     = {12:1--12:32},
  year      = {2019}
}

@String{EC       = {Conf.\ Economics and Computation (EC)}}

@article{geb/AzizB22,
  author       = {Haris Aziz and
                  Florian Brandl},
  title        = {The vigilant eating rule: {A} general approach for probabilistic economic
                  design with constraints},
  journal      = {Games Econ. Behav.},
  volume       = {135},
  pages        = {168--187},
  year         = {2022}
}

@article{foley1967resource,
	author    = {Duncan Foley},
	title     = {Resource Allocation and the Public Sector},
	journal   = {Yale Economic Essays},
	pages     = {45--98},
	year      = {1967}
}

@article{steihaus1948problem,
	title={The problem of fair division},
	author={Steinhaus, Hugo},
	journal={Econometrica},
	volume={16},
	pages={101--104},
	year={1948}
}

@article{aziz2021fair,
title = {Fair allocation of indivisible goods and chores},
author = {Aziz, Haris and Caragiannis, Ioannis and Igarashi, Ayumi and Walsh, Toby},
journal = {Autonomous Agents and Multi-Agent Systems},
year = {2021},
volume = {36},
number = {1},
doi = {10.1007/s10458-021-09532-8}
}

@inproceedings{bhaskar2021approximate,
title = {On Approximate Envy-Freeness for Indivisible Chores and Mixed Resources},
author = {Bhaskar, Umang and Sricharan, A. R. and Vaish, Rohit},
booktitle = {APPROX},
year = {2021},
doi = {10.4230/LIPIcs.APPROX/RANDOM.2021.1}
}

@inproceedings{Barman18FFEA,
 author = {Barman, Siddharth and Krishnamurthy, Sanath Kumar and Vaish, Rohit},
 title = {Finding Fair and Efficient Allocations},
 booktitle = {Proceedings of the 19th ACM Conference on Economics and Computation (EC)},
  year = {2018},
 pages = {557--574}
 }

@article{orl/AzizMS20,
  author    = {Haris Aziz and
               Herv{\'{e}} Moulin and
               Fedor Sandomirskiy},
  title     = {A polynomial-time algorithm for computing a Pareto optimal and almost
               proportional allocation},
  journal   = {Oper. Res. Lett.},
  volume    = {48},
  number    = {5},
  pages     = {573--578},
  year      = {2020}
}

@article{ai/AzizLMWZ24,
  author       = {Haris Aziz and
                  Bo Li and
                  Herv{\'{e}} Moulin and
                  Xiaowei Wu and
                  Xinran Zhu},
  title        = {Almost proportional allocations of indivisible chores: Computation,
                  approximation and efficiency},
  journal      = {Artif. Intell.},
  volume       = {331},
  pages        = {104118},
  year         = {2024}
}

@article{teco/ChakrabortyISZ21,
  author    = {Mithun Chakraborty and
               Ayumi Igarashi and
               Warut Suksompong and
               Yair Zick},
  title     = {Weighted Envy-freeness in Indivisible Item Allocation},
  journal   = {{ACM} Trans. Economics and Comput.},
  volume    = {9},
  number    = {3},
  pages     = {18:1--18:39},
  year      = {2021}
}

@article{ipl/Suksompong25,
  author       = {Warut Suksompong},
  title        = {Weighted fair division of indivisible items: {A} review},
  journal      = {Inf. Process. Lett.},
  volume       = {187},
  pages        = {106519},
  year         = {2025}
}

@article{jair/LiuLSW24,
  author       = {Shengxin Liu and
                  Xinhang Lu and
                  Mashbat Suzuki and
                  Toby Walsh},
  title        = {Mixed Fair Division: {A} Survey},
  journal      = {J. Artif. Intell. Res.},
  volume       = {80},
  pages        = {1373--1406},
  year         = {2024}
}

@article{sigecom/AzizLMW22,
  author       = {Haris Aziz and
                  Bo Li and
                  Herv{\'{e}} Moulin and
                  Xiaowei Wu},
  title        = {Algorithmic fair allocation of indivisible items: a survey and new
                  questions},
  journal      = {SIGecom Exch.},
  volume       = {20},
  number       = {1},
  pages        = {24--40},
  year         = {2022}
}

@article{ai/AmanatidisABFLMVW23,
  author       = {Georgios Amanatidis and
                  Haris Aziz and
                  Georgios Birmpas and
                  Aris Filos{-}Ratsikas and
                  Bo Li and
                  Herv{\'{e}} Moulin and
                  Alexandros A. Voudouris and
                  Xiaowei Wu},
  title        = {Fair division of indivisible goods: Recent progress and open questions},
  journal      = {Artif. Intell.},
  volume       = {322},
  pages        = {103965},
  year         = {2023}
}

@inproceedings{sagt/HalpernS19,
  author    = {Daniel Halpern and
               Nisarg Shah},
  title     = {Fair Division with Subsidy},
  booktitle = {{SAGT}},
  series    = {Lecture Notes in Computer Science},
  volume    = {11801},
  pages     = {374--389},
  publisher = {Springer},
  year      = {2019}
}

@inproceedings{ifaamas/ElmalemGS25,
  author       = {Noga Klein Elmalem and
                  Rica Gonen and
                  Erel Segal{-}Halevi},
  title        = {Weighted Envy Freeness With Bounded Subsidies},
  booktitle    = {{AAMAS}},
  pages        = {2504--2506},
  publisher    = {International Foundation for Autonomous Agents and Multiagent Systems
                  / {ACM}},
  year         = {2025}
}

@article{corr/abs-2502-09006,
  author       = {Noga Klein Elmalem and
                  Haris Aziz and
                  Rica Gonen and
                  Xin Huang and
                  Kei Kimura and
                  Indrajit Saha and
                  Erel Segal{-}Halevi and
                  Zhaohong Sun and
                  Mashbat Suzuki and
                  Makoto Yokoo},
  title        = {Whoever Said Money Won't Solve All Your Problems? Weighted Envy-free
                  Allocation with Subsidy},
  journal      = {CoRR},
  volume       = {abs/2502.09006},
  year         = {2025}
}

@inproceedings{ifaamas/0001HKS0SY25,
  author       = {Haris Aziz and
                  Xin Huang and
                  Kei Kimura and
                  Indrajit Saha and
                  Zhaohong Sun and
                  Mashbat Suzuki and
                  Makoto Yokoo},
  title        = {Weighted Envy-free Allocation with Subsidy},
  booktitle    = {{AAMAS}},
  pages        = {2417--2419},
  publisher    = {International Foundation for Autonomous Agents and Multiagent Systems
                  / {ACM}},
  year         = {2025}
}

@article{corr/LiSSX25,
  author       = {Bo Li and
                  Ankang Sun and
                  Mashbat Suzuki and
                  Shiji Xing},
  title        = {On the Subsidy of Envy-Free Orientations in Graphs},
  journal      = {CoRR},
  volume       = {abs/2502.13671},
  year         = {2025}
}

@article{corr/DaiCWXZ24,
  author       = {Sijia Dai and
                  Yankai Chen and
                  Xiaowei Wu and
                  Yicheng Xu and
                  Yong Zhang},
  title        = {Weighted Envy-Freeness in House Allocation},
  journal      = {CoRR},
  volume       = {abs/2408.12523},
  year         = {2024}
}

@article{knaster1929,
  title={Ein Beweis des Fixpunktsatzes f{\"u}r n-dimensionale Simplexe},
  author={Knaster, Bronis{\l}aw and Kuratowski, Kazimierz and Mazurkiewicz, Stefan},
  journal={Fundamenta Mathematicae},
  volume={14},
  number={1},
  pages={132--137},
  year={1929},
  publisher={Polska Akademia Nauk. Instytut Matematyczny PAN}
}

@inproceedings{ijcai/WuXZ25,
  author       = {Xiaowei Wu and
                  Quan Xue and
                  Shengwei Zhou},
  title        = {Revisiting Proportional Allocation with Subsidy: Simplification and
                  Improvements},
  booktitle    = {{IJCAI}},
  pages        = {4082--4090},
  publisher    = {ijcai.org},
  year         = {2025}
}

@inproceedings{garg2025wprop+PO4chores,
author = {Garg, Jugal and Sharma, Eklavya and Wu, Xiaowei},
title = {Proportional and Pareto-Optimal Allocation of Chores with Subsidy},
booktitle = {Autonomous Agents and Multiagent Systems (AAMAS)},
year = {2026},
pages = {2196–2204},
doi = {10.65109/UYKY7130}
}

@inproceedings{wine/WuZ24,
  author       = {Xiaowei Wu and
                  Shengwei Zhou},
  title        = {Tree Splitting Based Rounding Scheme for Weighted Proportional Allocations
                  with Subsidy},
  booktitle    = {{WINE}},
  series       = {Lecture Notes in Computer Science},
  volume       = {15534},
  pages        = {295--313},
  publisher    = {Springer},
  year         = {2024}
}

@inproceedings{wine/WuZZ23,
  author       = {Xiaowei Wu and
                  Cong Zhang and
                  Shengwei Zhou},
  title        = {One Quarter Each (on Average) Ensures Proportionality},
  booktitle    = {{WINE}},
  series       = {Lecture Notes in Computer Science},
  volume       = {14413},
  pages        = {582--599},
  publisher    = {Springer},
  year         = {2023}
}

@inproceedings{sigecom/BrustleDNSV20,
  author       = {Johannes Brustle and
                  Jack Dippel and
                  Vishnu V. Narayan and
                  Mashbat Suzuki and
                  Adrian Vetta},
  title        = {One Dollar Each Eliminates Envy},
  booktitle    = {{EC}},
  pages        = {23--39},
  publisher    = {{ACM}},
  year         = {2020}
}

@inproceedings{corr/abs-2507-03946,
  author       = {Siddharth Barman and
                  Vishwa Prakash HV and
                  Aditi Sethia and
                  Mashbat Suzuki},
  title        = {Fair and Efficient Allocation of Indivisible Mixed Manna},
  booktitle = {WINE},
  year         = {2025}
}

@article{ai/WuZZ25,
  author       = {Xiaowei Wu and
                  Cong Zhang and
                  Shengwei Zhou},
  title        = {Weighted {EF1} allocations for indivisible chores},
  journal      = {Artif. Intell.},
  volume       = {347},
  pages        = {104386},
  year         = {2025}
}

@inproceedings{corr/abs-2507-09544,
  author       = {Ryoga Mahara},
  title        = {Existence of Fair and Efficient Allocation of Indivisible Chores},
  booktitle      = {SODA},
  year         = {2026}
}

@article{corr/abs-2509-18673,
  author       = {Siddharth Barman and
                  Paritosh Verma},
  title        = {Proximately Envy-Free and Efficient Allocation of Mixed Manna},
  journal      = {CoRR},
  volume       = {abs/2509.18673},
  year         = {2025}
}

@article{branzei2024algorithms,
title = {Algorithms for Competitive Division of Chores},
author = {Br{\^a}nzei, Simina and Sandomirskiy, Fedor},
journal = {Mathematics of Operations Research},
year = {2024},
volume = {49},
pages = {398--429},
publisher = {INFORMS},
doi = {10.1287/moor.2023.1361}
}

@book{brams1996fair,
title = {Fair Division: From Cake-Cutting to Dispute Resolution},
author = {Brams, S.J. and Taylor, A.D.},
year = {1996},
publisher = {Cambridge University Press},
url = {https://books.google.com/books?id=cLUA-sRhJ5QC}
}
\bibliographystyle{alpha}

\newpage

\appendix
\crefalias{section}{appendix}
\crefalias{subsection}{appendix}
\crefalias{subsubsection}{appendix}

\section{Missing Proofs from Section~\ref{sec:redn-canon} and Section~\ref{sec:tight-subsidy}}
\label{sec:missing_proofs}

\begin{proofof}{Lemma~\ref{lemma:constructed_instance_is_canonical}}
    We first show that the perturbation does not change the sign of any value $\hat{v}_i(e)$ if it is non-zero.
    For any item $e$ and agent $i$, if $\hat{v}_i(e) < 0$, then we also have $v'_i(e) < 0$ because $\eta\leq\lambda<|\hat{v}_i(e)|$, by the partial-welfare separation condition.
    If $\hat{v}_i(e) \geq 0$, then we have $v'_i(e) > 0$ since the noises are positive.
    Therefore, the resulting instance $\cI' = (N', M', v', w')$ is strictly objective.

    The instance is also consistent because for any agent $i$,
    \begin{itemize}
        \item if $\hat{v}_i(M') \geq 0$, then we have $v'_i(M') > 0$;
        \item if $\hat{v}_i(M') < 0$, then we have $v'_i(M') < 0$ because the total noise is strictly less than $\lambda<|\hat{v}_i(M')|$, again by the partial-welfare separation condition.
    \end{itemize}

    Finally, fix a simple alternating cycle as in Definition~\ref{defn:degen}. If its two products are equal under $v'$, then
    \begin{equation*}
        \eta_{i_0 e_1} = \frac{\prod_{t=1}^k  v'_{i_t}(e_t)}{\prod_{t=2}^k  v'_{i_{t-1}}(e_t)} - \hat{v}_{i_0}(e_1),
    \end{equation*}
    which happens with probability $0$ conditioned on any realization of the random noises of other values.
    Since the number of cycles is finite, a union bound implies that the instance $\cI'$ is non-degenerate with probability $1$.
\end{proofof}

\begin{proofof}{Lemma~\ref{lemma:reduction_to_canonical_instances}}
    Given any allocation $\bX'$ for $\cI'$, we construct an allocation $\bX$ for $\cI$ as follows.
    We initialize $X_i = X'_i$ for all $i\in N'$ and $X_i = \varnothing$ for all $i\in N\setminus N'$.
    For any $e\in X_i$, if $v_i(e) < 0$ but $v'_i(e) > 0$, then we know that during the construction of instance $\cI'$, $v_i(e)$ was first artificially set to $0$, and then a positive random noise was added.
    In this case, there must exist $j\in N$ such that $v_j(e) > 0$.
    We reallocate $e$ from $X_i$ to $X_j$, which increases both $v_i(X_i)$ and $v_j(X_j)$.
    Note that after these reallocations, for all $e\in X_i\cap M'$ we have $v_i(e) = \hat{v}_i(e)$.
    Finally, for all $e\in M\setminus M'$ (which are the items satisfying that $v_i(e)\leq 0$ for all $i$), we allocate $e$ to an arbitrary agent $i\in N$ with $v_i(e) = 0$.
    Clearly, $\bX$ is a complete allocation for instance $\cI$.

    Every removed agent $i\in N\setminus N'$ receives only nonnegative-value items and satisfies $v_i(M)\leq\hat{v}_i(M')<0$, so she is WPROP without subsidy.
    Fix any remaining agent $i\in N'$.
    We have
    \begin{equation*}
        v_i(X_i) = \hat{v}_i(X_i\cap M') \geq \hat{v}_i(X'_i) \geq v'_i(X'_i) - \eta\cdot |X'_i|,
    \end{equation*}
    where the first inequality holds because all items added to $X'_i$ have nonnegative value, and every item removed from $X'_i$ has zero value under $\hat{v}_i$. Also, since the total perturbation on any bundle is strictly less than $\lambda$, we have $v_i(X_i)>v'_i(X'_i)-\lambda$.
    Recall that during the construction of instance $\cI'$, we did not decrease the value of any agent on any item.
    Therefore, we have $w'_i\cdot v'_i(M') \geq w_i\cdot v_i(M)$ because:
    \begin{itemize}
        \item if $N'\neq N$, then $w'_i\geq w_i$ and $v'_i(M')>\hat{v}_i(M')\geq0$, with $\hat{v}_i(M')\geq v_i(M)$;
        \item if $N'=N$, then $w'_i=w_i$ and $v'_i(M')>v_i(M)$.
    \end{itemize}

    In summary, we have
    \begin{align*}
        \subsidy(\bX, \cI) &= \sum_{i\in N} \left( w_i\cdot v_i(M) - v_i(X_i) \right)^+
        \leq \sum_{i\in N'} \left( w'_i\cdot v'_i(M') - v'_i(X'_i) + \eta\cdot |X'_i| \right)^+
        \\ & \leq \subsidy(\bX', \cI') + \eta\cdot m
        \leq \subsidy(\bX', \cI') + \lambda.
    \end{align*}

    By choosing an arbitrarily small $\lambda$, the difference in subsidy is negligible.

    It remains to show that allocation $\bX$ is WPROP1 for instance $\cI$ if $\bX'$ is WPROP1 for instance $\cI'$.
    Removed agents are already WPROP, so fix $i\in N'$. By the definition of WPROP1, we have three cases.
    \begin{itemize}
        \item Suppose $v'_i(X'_i) \geq w'_i\cdot v'_i(M')$. Then we have $v_i(X_i) \geq w_i\cdot v_i(M)$ because otherwise by definition of $\lambda$ and $\eta$, we get
        \begin{equation*}
            v_i(X_i) \leq w_i\cdot v_i(M) - \lambda
            \leq w'_i\cdot v'_i(M') - \lambda
            \leq v'_i(X'_i) - \lambda,
        \end{equation*}
        contradicting $v_i(X_i)>v'_i(X'_i)-\lambda$.

        \item Suppose there exists a good $e\in M'\setminus X'_i$ (with $v'_i(e) > 0$) such that $v'_i(X'_i + e) \geq w'_i\cdot v'_i(M')$.
        If $e\notin M\setminus X_i$, then we have $e\in X_i$, i.e., $e$ is added to $X_i$ during the construction of $\bX$.
        Note that in this case we must have $v_i(e) > 0$, and thus $v_i(e) > v'_i(e) - \eta$.
        Then we have
        \begin{equation*}
            v_i(X_i) > v'_i(X'_i + e) - \eta\cdot (|X'_i|+1)
            \geq w_i\cdot v_i(M) - \lambda,
        \end{equation*}
        which (by definition of $\lambda$) implies $v_i(X_i) \geq w_i\cdot v_i(M)$.
        If $e\in M\setminus X_i$, then
        \begin{itemize}
            \item If $v_i(e) > 0$, then the above lower bounds apply to $v_i(X_i + e)$, which implies WPROP1.
            \item If $v_i(e) \leq 0$, then we have $v'_i(e) < \eta$, and the above lower bounds still apply to $v_i(X_i)$, which implies WPROP.
        \end{itemize}

        \item Suppose there exists a chore $e\in X'_i$ (with $v'_i(e) < 0$) such that $v'_i(X'_i - e) \geq w'_i\cdot v'_i(M')$.
        Note that in this case, we must have $e\in X_i$, as all items removed from $X'_i$ during the construction of $\bX$ have positive value under $v'_i$.
        Since $v_i(e)=\hat{v}_i(e)$ and the remaining bundle contains at most $m-1$ items, we have
        \begin{equation*}
            v_i(X_i - e) \geq v'_i(X'_i - e) - \eta\cdot (|X'_i|-1)
            > w_i\cdot v_i(M) - \lambda,
        \end{equation*}
        which implies $v_i(X_i - e) \geq w_i\cdot v_i(M)$.
    \end{itemize}

    Therefore, the allocation $\bX$ is WPROP1 for instance $\cI$.
\end{proofof}

\begin{proofof}{Lemma~\ref{lemma:closedness_of_C}}
Let $B = (\beta^{(1)}, \beta^{(2)}, \ldots)$ be any sequence
where $\beta^{(t)} \in C_i$ for every $t \in \N$.
We show that if $B$ converges to $\beta$, then $\beta \in C_i$.

For any $t \in \N$, $\beta^{(t)} \in C_i$ implies that
there exists an allocation $\bX^{(t)}$ such that $\bX^{(t)}$ is MBB-feasible for $\betahat^{(t)}$
and $i \in \argmin_{j \in N} \{\delta_j(\bX^{(t)}, \betahat^{(t)}) \}$.
Since the set of integral allocations is finite, there must exist an allocation $\bX$ that repeats infinitely often in the sequence $(\bX^{(t)})_{t \in \N}$.
Let $(t_1, t_2, \ldots)$ be the indices of $\bX$ in this sequence.
The subsequence $(\beta^{(t_1)}, \beta^{(t_2)}, \ldots)$ of $B$ also converges to $\beta$.
Thus, let us assume \wLoG{} that $\bX^{(t)} = \bX$ for all $t \in \N$.

Let $p^{(t)}(e) \defeq \max_{i \in N} \{ \betahat^{(t)}_i \cdot v_i(e) \}$.
Since $B$ converges to $\beta$, we get that the sequence $(p^{(t)})_{t \in \N}$
converges pointwise to $p(e) \defeq \max_{i \in N} \{ \betahat_i \cdot v_i(e) \}$.
Since $\betahat \ge \eps$ and $\betahat^{(t)} \ge \eps$ for all $t \in \N$,
for every $i\in N$, the sequence $(p^{(t)}(e)/\betahat^{(t)}_i)_{t \in \N}$ converges to $p(e)/\betahat_i$.
For any $t \in \N$, the allocation $\bX$ is MBB-feasible for $\betahat^{(t)}$,
i.e., for all $i \in N$ and $e \in M$, we have $x_{i,e} > 0 \implies p^{(t)}(e) = \betahat^{(t)}_i \cdot v_i(e)$.
Thus, we get that $x_{i,e} > 0 \implies p(e) = \betahat_i \cdot v_i(e)$,
so $\bX$ is MBB-feasible for $\betahat$ too.

We now show that $i \in \argmin_{j \in N} \{\delta_j(\bX, \betahat)\}$, which will establish that $\beta \in C_i$.
For all $t \in \N$, let
\begin{align*}
\phi^{(t)} &\defeq \betahat^{(t)}_i\cdot \left(w_i\cdot v_i(M) - v_i(X_i) \right),
\\ \theta^{(t)} &\defeq \min_{j \in N} \left\{ \betahat^{(t)}_j\cdot (w_j\cdot v_j(M) - v_j(X_j)) \right\}.
\end{align*}
Then the sequence $(\phi^{(t)})_{t \in \N}$ converges to $\phi$,
and the sequence $(\theta^{(t)})_{t \in \N}$ converges to $\theta$, where
\begin{align*}
\phi &\defeq \betahat_i\cdot \left(w_i\cdot v_i(M) - v_i(X_i) \right),
\\ \theta &\defeq \min_{j \in N}\left\{ \betahat_j\cdot \left(w_j\cdot v_j(M) - v_j(X_j)\right) \right\}.
\end{align*}
Moreover, $\phi^{(t)} \le \theta^{(t)}$ for all $t \in \N$,
since $i \in \argmin_{j \in N} \{\delta_j(\bX, \betahat^{(t)})\}$.
Thus, $\phi \le \theta$, and so $i \in \argmin_{j \in N} \{\delta_j(\bX, \betahat)\}$.
Hence, $\beta \in C_i$, and so, $C_i$ is closed.
\end{proofof}

\begin{proofof}{Lemma~\ref{lemma:no_cycle_in_MBB_graph}}
    Suppose there exists a simple alternating cycle $\left(i_0,e_1,i_1,e_2,\ldots,e_k,i_k\right)$, where $k\geq2$, $i_k=i_0$, the agents $i_0,\ldots,i_{k-1}$ are distinct, and the items $e_1,\ldots,e_k$ are distinct.
    Then by definition of MBB graph, for each item $e_t$ in the cycle, $e_t$ is an MBB item for both agent $i_{t-1}$ and agent $i_t$, which implies that
    \begin{equation*}
        p(e_t) = \beta_{i_{t-1}}\cdot v_{i_{t-1}}(e_t) = \beta_{i_{t}}\cdot v_{i_{t}}(e_t).
    \end{equation*}

    Therefore we have
    \begin{equation*}
        \prod_{t=1}^k \frac{v_{i_{t-1}}(e_t)}{v_{i_t}(e_t)} = \prod_{t=1}^k \frac{\beta_{i_{t}}}{\beta_{i_{t-1}}} = 1,
    \end{equation*}
    which contradicts non-degeneracy.
\end{proofof}

\section{Details on the Algorithm for a Few Types of Agents}
\label{sec:few-types-details}

\subsection{Greedy Algorithm for Identical Valuations}
\label{sec:greedy-details}

We describe here an algorithm for obtaining a low subsidy allocation when the agents have the same valuation functions,
and all the items are goods or all the items are chores.
We have $n$ agents and a set $M$ of items.
Each agent $i$ has weight $w_i > 0$.
The agents have a common valuation function $u$.

\paragraph{Step 1: Allocate most items.}
Create $n$ \emph{bins}, where the $i\Th$ bin has a capacity of $w_i\cdot |u(M)|$.
Each item $e \in M$ has a \emph{size} $|u(e)|$.
A set $S \subseteq M$ of items is said to \emph{fit} in a bin
if the sum of their sizes is at most the bin's capacity.
Now we repeatedly pick an item and place it in a bin if it fits,
and we stop when this operation is no longer feasible.
Let $B_i$ be the set of items in the $i\Th$ bin
and let $\Mhat \defeq M \setminus \bigcup_{i=1}^n B_i$ be the items that remain outside the bins.
Let $z_i \defeq w_i\cdot |u(M)| - |u(B_i)|$ be the remaining space in the $i\Th$ bin.
If $\Mhat=\varnothing$, every bin is exactly filled, so return $A_i=B_i$ for all $i$; no subsidy is needed.
First, note that
\[ \sum_{i=1}^n z_i = |u(M)| - \sum_{i=1}^n |u(B_i)| = |u(\Mhat)|. \]
Next, since no item can be placed in any bin, we have
$z_i < \min_{e \in \Mhat}\{ |u(e)| \}$ for all $i \in N$. Thus,
\begin{align*}
|u(\Mhat)| = \sum_{i=1}^n z_i < n\cdot \min_{e \in \Mhat}\{ |u(e)| \} \le n\cdot  \frac{|u(\Mhat)|}{|\Mhat|},
\end{align*}
which implies that $|\Mhat| < n$.

\paragraph{Step 2: Allocate the remaining items.}
Assume \wLoG{} that $z_1 \ge \cdots \ge z_n$. Let $\Mhat = \{e_1, \ldots, e_t\}$.
We assign an item from $\Mhat$ to each of the first $t$ agents.
Specifically, we output the allocation $A = (A_1, \ldots, A_n)$, where
\[ A_i = \begin{cases}B_i \cup \{e_i\} & \text{ if } i \le t \\ B_i & \text{ otherwise}\end{cases}. \]

\begin{lemma}
The total subsidy required to make allocation $A$ proportional is at most $\tau(n)$.
\end{lemma}
\begin{proof}
\textbf{Case 1}: All items are goods.
Then the first $t$ agents are PROP-satisfied, and so, require no subsidy.
Every other agent $i$ requires a subsidy of $z_i$.
Thus, the total subsidy required is $\sum_{i=t+1}^n z_i$.

\textbf{Case 2}: All items are chores.
Then the last $n-t$ agents are PROP-satisfied, and so, require no subsidy.
Every other agent $i$ requires a subsidy of $|u(e_i)| - z_i$.
Thus, the total subsidy is
\[ \sum_{i=1}^t (|u(e_i)| - z_i) = |u(\Mhat)| - \sum_{i=1}^t z_i = \sum_{i=t+1}^n z_i. \]

Recall that $z_1 \ge \cdots \ge z_n$.
Thus, the subsidy required in both cases is
\[ \sum_{i=t+1}^n z_i \le \frac{n-t}{n}\cdot \sum_{i=1}^n z_i = \frac{n-t}{n}\cdot |u(\Mhat)|
    \le \frac{t\cdot (n-t)}{n} \le \tau(n).
\qedhere \]
\end{proof}

\paragraph{Running time.}
Converting the instance to one containing only goods or only chores takes $O(m)$ time.
In step 1, we can pack an item into a bin if and only if the smallest item
can fit into the bin with the largest empty space.
Such queries can be answered by maintaining balanced binary search trees for items and bins.
Thus, step 1 can be performed in $O(m\log m + n\log n + m\log n)$ time.
Step 2 can be performed in $O(m\log n)$ time.
Thus, the total running time of this algorithm is $O((m+n)\log(m+n))$.

\subsection{Enumerating fPO Allocations}
\label{sec:fpo-enum-details}

\subsubsection{Two Agents}
\label{sec:fpo-enum-details:n2}

For an adjusted-winner-like approach to work, we require each item $e$ to have
a distinct $|v_{i}(e)|/|v_{j}(e)|$.
Fortunately, this is a simple consequence of non-degeneracy.
Note that while the items contain both goods and chores, for strictly objective instances, we always have $|v_{i}(e)|/|v_{j}(e)| = v_{i}(e)/v_{j}(e) > 0$.

\begin{observation}
\label{thm:semi-non-degen}
Let $\Ical \defeq (N, M, v, w)$ be a strictly objective non-degenerate fair division instance.
Then for any two distinct agents $i, j \in N$, and any two distinct items $e_1, e_2 \in M$, we have
\begin{equation*}
    \frac{v_{i}(e_1)}{v_{j}(e_1)} \neq \frac{v_{i}(e_2)}{v_{j}(e_2)}.
\end{equation*}
\end{observation}

We now define $\sortAndSplit$, a method of splitting items between two agents.

\begin{definition}[$\sortAndSplit$]
\label{defn:sort-and-split}
Let $\Ical \defeq (N, M, v, w)$ be a strictly objective non-degenerate fair division instance,
where $M^+ \subseteq M$ is the set of goods and $M^- \subseteq M$ is the set of chores.
Let $i,j\in N$ be distinct and $M'\subseteq M$, and write $m'=|M'|$.
If $M'=\varnothing$, define $\sortAndSplit_{\Ical}\left(i,j,M',0\right)=\left(\varnothing,\varnothing\right)$. Otherwise, proceed as follows.
Let $\sigma = (\sigma_1, \ldots, \sigma_{m'})$ be the ordering of $M'$ such that
\begin{equation*}
    \frac{v_{i}(\sigma_1)}{v_{j}(\sigma_1)} > \cdots
    > \frac{v_{i}(\sigma_{m'})}{v_{j}(\sigma_{m'})}.
\end{equation*}
For any $k\in\left\{0,1,\ldots,m'\right\}$, with $[0]=\varnothing$, let $L^{(k)} \defeq \{\sigma_t: t \in [k]\}$ be the first $k$ items in $\sigma$
and $R^{(k)} \defeq \{\sigma_t: t \in [m'] \setminus [k]\} = M'\setminus L^{(k)}$ be the last $m'-k$ items in $\sigma$.
Define $\sortAndSplit_{\Ical}(i, j, M', k)$ as the pair $(A_1, A_2)$, where
\begin{equation*}
    A_1 \defeq (L^{(k)} \cap M^+) \cup (R^{(k)} \cap M^-) \quad \text{and}
    \quad A_2 \defeq (L^{(k)} \cap M^-) \cup (R^{(k)} \cap M^+).
\end{equation*}
\end{definition}

Note that in the above definition, $(A_1, A_2)$ is a partition of $M'$.
We prove that every fPO allocation between two agents is an output of $\sortAndSplit$.

\begin{lemma}
\label{thm:n2-fpo-is-sas}
Let $\Ical \defeq ([2], M, (v_i)_{i=1}^2, w)$ be a strictly objective non-degenerate fair division instance.
Let $X = (X_1, X_2)$ be an fPO allocation.
Then there exists $k \in \{0,1,\ldots,|M|\}$ such that
$(X_1, X_2) = \sortAndSplit_{\Ical}(1, 2, M, k)$.
\end{lemma}
\begin{proof}
Let $M^+$ be the set of goods, $M^-$ be the set of chores, and $m \defeq |M|$. W.l.o.g., assume that
\begin{equation*}
    \frac{v_1(1)}{v_2(1)} > \cdots > \frac{v_1(m)}{v_2(m)}.
\end{equation*}
For each item $e \in M$, we give it 
\begin{itemize}
    \item a label of `1' if $e \in X_1 \cap M^+$ or $e \in X_2 \cap M^-$,
    \item a label of `2' if $e \in X_1 \cap M^-$ or $e \in X_2 \cap M^+$.
\end{itemize}

Let $k$ be the number of items labeled 1.
We now show that all items in $[k]$ are labeled 1, and the rest are labeled 2, which implies that $(X_1, X_2) = \sortAndSplit_{\Ical}(1, 2, M, k)$.

Suppose there exist items $e_1 < e_2$ such that $e_1$ is labeled 2 and $e_2$ is labeled 1.
We now show how to transfer some amount of $e_1$ and $e_2$ between the agents
to increase the utility of both agents, which contradicts $X$ being fPO.
For all $e \in M$, define 
\begin{equation*}
    \rho_e \defeq \frac{v_1(e)}{v_2(e)} \quad \text{and} \quad \beta_e \defeq \sqrt{v_1(e)\cdot v_2(e)}.
\end{equation*}

Then we have $\rho_1 > \cdots > \rho_m$.
Let $0 < \eps \le \min\left\{\beta_{e_1},\beta_{e_2}\right\}$.

First, let us change the owner of an $\eps/\beta_{e_1}$ fraction of $e_1$.
Recall that $e_1$ has label 2.
If $e_1 \in M^+$, this involves transferring a fraction of $e_1$ from agent 2 to agent 1,
and if $e_1 \in M^-$, this involves transferring a fraction of $e_1$ from agent 1 to agent 2.
In either case, agent 1's utility increases by
\begin{equation*}
    |v_1(e_1)|\cdot \frac{\eps}{\beta_{e_1}} = \sqrt{\rho_{e_1}}\cdot \eps,
\end{equation*}
and agent 2's utility decreases by
\begin{equation*}
    |v_2(e_1)|\cdot \frac{\eps}{\beta_{e_1}} = \frac{\eps}{\sqrt{\rho_{e_1}}}.
\end{equation*}

Next, let us change the owner of an $\eps/\beta_{e_2}$ fraction of $e_2$.
If $e_2 \in M^+$, this involves transferring a fraction of $e_2$ from agent 1 to agent 2,
and if $e_2 \in M^-$, this involves transferring a fraction of $e_2$ from agent 2 to agent 1.
In either case, agent 1's utility decreases by 
\begin{equation*}
    |v_1(e_2)|\cdot \frac{\eps}{\beta_{e_2}} = \sqrt{\rho_{e_2}}\cdot \eps,
\end{equation*}
and agent 2's utility increases by
\begin{equation*}
    |v_2(e_2)|\cdot \frac{\eps}{\beta_{e_2}} = \frac{\eps}{\sqrt{\rho_{e_2}}}.
\end{equation*}

Since $\rho_{e_1} > \rho_{e_2}$, agent 1's net utility increase is $(\sqrt{\rho_{e_1}} - \sqrt{\rho_{e_2}})\cdot \eps > 0$,
and agent 2's net utility increase is $\left(\frac{1}{\sqrt{\rho_{e_2}}} - \frac{1}{\sqrt{\rho_{e_1}}}\right)\cdot \eps > 0$.
Thus, the new allocation obtained by making these transfers strictly increases both agents' utilities, which contradicts $X$ being fPO.
\end{proof}

\subsubsection{More Than Two Agents}
\label{sec:fpo-enum-details:n2+}

We first make the simple observation that in any fPO allocation,
the sub-allocation between any pair of agents is also fPO.

\begin{lemma}
\label{thm:fpo-restr-2}
Let $A$ be an fPO allocation for the fair division instance
$\Ical \defeq (N, M, v, w)$.
Let $i, j \in N$.
Let $N' \defeq \{i, j\}$ and $M' \defeq A_{i} \cup A_{j}$.
Let $\Ical' \defeq (N', M', (v_i, v_j), (\frac{w_i}{w_i+w_j}, \frac{w_j}{w_i+w_j}))$.
Then $(A_i, A_j)$ is an fPO allocation for the instance $\Ical'$.
\end{lemma}
\begin{proof}
Suppose a fractional allocation $\bx$ Pareto-dominates $(A_i,A_j)$ in $\Ical'$.
Let $\by$ be a fractional allocation for $\Ical$ where $y_i=x_i$, $y_j=x_j$ on $M'$, with zero fractions for these agents on $M\setminus M'$,
and for all $t \in N \setminus N'$ and $e \in M$, we have $y_{t,e} \defeq \boolOne(e \in A_t)$.
Then $\by$ Pareto-dominates $A$ in $\Ical$, which contradicts $A$ being fPO.
\end{proof}

By the above lemma and our previous analysis on two agents, we show that any fPO allocation can be obtained by the following process:
make each pair of agents \emph{fight over} the items using $\sortAndSplit$ (by enumerating $k$),
and each agent's bundle is based on the items she \emph{wins} in every fight.

\begin{definition}[$\condorcetSplit$]
\label{defn:condorcet-split}
Let $\Ical \defeq (N, M, v, w)$ be a strictly objective non-degenerate fair division instance.
Let $P_N \defeq \{(i, j): i \in N, j \in N, i < j\}$
be the set of all pairs of agents.
For any $(i, j) \in P_N$, let $k_{i, j} \in \{0,1,\ldots,|M|\}$,
and $(A^{(i, j)}_{i}, A^{(i, j)}_{j}) \defeq \sortAndSplit(i, j, M, k_{i, j})$.
For any agent $t \in N$, define
\[ A_t \defeq \left(\bigcap_{t' \in [t-1]} A^{(t', t)}_t\right)
    \cap \left(\bigcap_{t' \in N \setminus [t]} A^{(t, t')}_t\right). \]
Let $\mathbf{k} = (k_{i,j})_{(i,j) \in P_N}$.
Then $\condorcetSplit_{\Ical}(\mathbf{k})$ is defined as the sequence $(A_1, \ldots, A_n)$.
\end{definition}

Clearly, the sets $(A_1,A_2,\ldots,A_n)$ in the output of $\condorcetSplit$ are pairwise-disjoint, because if an item $e$ is allocated to both agents $i$ and $j$, then it means that both agents win the item in $\sortAndSplit_{\Ical}(i, j, M, k_{i,j})$, which is impossible.

\begin{lemma}
\label{thm:condorcet-is-alloc}
The sets $(A_1, \ldots, A_n) \defeq \condorcetSplit_{\Ical}(\mathbf{k})$ are pairwise-disjoint.
\end{lemma}

However, the output of $\condorcetSplit$ need not be a full allocation, i.e., some items might be unallocated.
For example, for $n = 3$, $k_{1,2} = 0$, $k_{2,3} = 0$, and $k_{1,3} = m$, all sets output by $\condorcetSplit(\mathbf{k})$ are empty.
In the following, we show that every fPO allocation
is the output of $\condorcetSplit(\mathbf{k})$ for some vector $\mathbf{k}$, which implies that by identifying the full and fPO allocations that are output by $\condorcetSplit(\mathbf{k})$, we can enumerate all fPO allocations.

\begin{lemma}
\label{thm:fpo-pair-is-sas}
Let $\Ical \defeq (N, M, v, w)$ be a strictly objective non-degenerate fair division instance.
Let $\bX$ be an fPO allocation. Then for any two agents $i < j$,
there exists $k_{i,j} \in \{0,1,\ldots,|M|\}$ such that
$X_{i} \subseteq A^{(i, j)}_{i}$ and $X_{j} \subseteq A^{(i, j)}_{j}$, where
$(A^{(i, j)}_{i}, A^{(i, j)}_{j}) \defeq \sortAndSplit(i, j, M, k_{i, j})$.
\end{lemma}
\begin{proof}
Let $m=|M|$ and $\sigma = (\sigma_1, \ldots, \sigma_m)$ be an ordering of the items such that
\begin{equation*}
    \frac{v_{i}(\sigma_1)}{v_{j}(\sigma_1)} > \cdots > \frac{v_{i}(\sigma_m)}{v_{j}(\sigma_m)}.
\end{equation*}

Let $\Ical'$ be the instance induced by $N' \defeq \{i, j\}$ and $M' \defeq X_{i} \cup X_{j}$.
Then $(X_i, X_j)$ is an fPO allocation for the instance $\Ical'$ (by \cref{thm:fpo-restr-2}).
By \cref{thm:n2-fpo-is-sas},
there exists $k$ such that $(X_{i},X_{j}) =
\mathrm{sortAndSplit}(i, j, M',k)$.
Note that the ordering of items in $\mathrm{sortAndSplit}(i, j, M',k)$ is a subsequence of $\sigma = (\sigma_1, \ldots, \sigma_m)$, which implies that there exists some $k'$ such that $X_{i} \subseteq A^{(i, j)}_{i}$ and $X_{j} \subseteq A^{(i, j)}_{j}$, where
$(A^{(i, j)}_{i}, A^{(i, j)}_{j}) \defeq \sortAndSplit(i, j, M, k')$.
\end{proof}

\begin{lemma}
\label{thm:fpo-is-sas}
Let $\Ical \defeq (N, M, v, w)$ be a strictly objective non-degenerate fair division instance.
Let $\bX$ be an fPO allocation. Then $\bX = \condorcetSplit_{\Ical}(\mathbf{k})$ for some vector $\mathbf{k}$.
\end{lemma}
\begin{proof}
By \cref{thm:fpo-pair-is-sas}, for any two agents $i < j$,
there exists $k_{i,j} \in \{0,1,\ldots,|M|\}$ such that
$X_{i} \subseteq A^{(i, j)}_{i}$ and $X_{j} \subseteq A^{(i, j)}_{j}$, where
$(A^{(i, j)}_{i}, A^{(i, j)}_{j}) \defeq \sortAndSplit(i, j, M, k_{i, j})$.
Recall that $P_N$ is the set of all pairs of agents (as defined in \cref{defn:condorcet-split}).
Let $\mathbf{k} = (k_{i,j})_{(i,j) \in P_N}$.
Define $A \defeq \condorcetSplit_{\Ical}(\mathbf{k})$.
Then for any $i \in N$, we have $X_i \subseteq A_i$.
Moreover, the sets $A_1, \ldots, A_n$ are pairwise-disjoint by \cref{thm:condorcet-is-alloc}.

Suppose $X_i \neq A_i$ for some $i \in N$. Let $e \in A_i \setminus X_i$.
Since $\bX$ is a full allocation, there exists $j \in N$ such that $e \in X_j$, and so $e \in A_j$.
Thus, $A_i$ and $A_j$ are not disjoint, which is a contradiction.
Hence, we have $X_i = A_i$ for all $i \in N$ and $\bX = \condorcetSplit_{\Ical}(\mathbf{k})$.
\end{proof}

Finally, we show that given any allocation $A$ (e.g., an output of $\condorcetSplit_{\Ical}(\mathbf{k})$), we can decide whether it is fPO in polynomial time by solving a linear program.
Note that the same method applies to fractional allocations.

\begin{observation}
\label{thm:check-fpo-using-lp}
A fractional allocation $\by$ is fPO for the fair division instance
$(N, M, v, w)$ iff the optimal objective value
of the following linear program is at most the objective value of $\by$:
\begin{align*}
    \max_{\bx \in \mathbb{R}_{\ge 0}^{n \times m}} \quad
    &  \sum_{i\in N}\sum_{e\in M} v_i(e)\cdot x_{i}(e) & \\
    \text{subject to} \quad
    &  \sum_{e\in M} v_i(e)\cdot x_{i}(e) \ge \sum_{e\in M} v_i(e)\cdot y_{i}(e)
    & \text{for all } i \in N \\
    \text{and} \quad
    &  \sum_{i\in N} x_{i}(e) = 1  & \text{for all } e \in M
\end{align*}
\end{observation}

\subsection{Two-Stage Allocation}
\label{sec:two-stage-extra}

\begin{proofof}{Lemma~\ref{thm:hier-subsidy}}
For each group $j \in [k]$, let $S_j \defeq W_j \cdot u_j(M) - u_j(\Ahat_j)$.
Then
\[ \subsidy(\Ahat, \Icalhat) = \sum_{j=1}^k (S_j)^+. \]

For each $j \in [k]$, we have
\begin{align*}
\subsidy\left(A^{(j)}, \Ical^{(j)}\right) &= \sum_{i \in N_j} \left(\frac{w_i}{W_j}\cdot u_j(\Ahat_j) - u_j(A_i^{(j)})\right)^+ \\
& = \sum_{i \in N_j} \left(w_i \cdot u_j(M) - u_j(A_i^{(j)}) - \frac{w_i}{W_j}\cdot S_j\right)^+.
\end{align*}

Thus,
\begin{align*}
\subsidy(A, \Ical) &= \sum_{j=1}^k \sum_{i \in N_j} \left(w_i\cdot  u_j(M) - u_j(A_i^{(j)}) \right)^+ \\ 
& = \sum_{j=1}^k \sum_{i \in N_j} \left(\left(w_i\cdot  u_j(M) - u_j(A_i^{(j)}) - \frac{w_i}{W_j}\cdot S_j\right)
    + \frac{w_i}{W_j}\cdot S_j\right)^+ \\ 
& \le \sum_{j=1}^k \sum_{i \in N_j} \left(\left(w_i\cdot u_j(M) - u_j(A_i^{(j)}) - \frac{w_i}{W_j}\cdot S_j\right)^+
    + \frac{w_i}{W_j}\cdot (S_j)^+\right) \\ 
& = \sum_{j=1}^k \left( \subsidy(A^{(j)}, \Ical^{(j)}) + (S_j)^+ \right) \\ 
& = \subsidy(\Ahat, \Icalhat) + \sum_{j=1}^k \subsidy( A^{(j)}, \Ical^{(j)} ).
\end{align*}
\end{proofof}

\end{document}